\documentclass{LMCS}

\def\dOi{9(4:4)2013}
\lmcsheading%
{\dOi}
{1--23}
{}
{}
{Nov.~14, 2012}
{Oct.~11, 2013}
{}

\ACMCCS{[{\bf Theory of computation}]: Semantics and reasoning---Program semantics}

\usepackage{microtype}
\usepackage{amsmath,amssymb}
\usepackage{nicefrac}
\usepackage{xspace}
\usepackage{stmaryrd}
\usepackage{mathpartir}
\usepackage{hyperref}

\theoremstyle{remark}

\newcommand{\N}{\ensuremath{\mathbb{N}}\xspace}

\newcommand{\ordURel}[1]{\ensuremath{\textit{Rel}_{#1}}\xspace} 
\newcommand{\VRel}{\textit{VRel}\xspace} 
\newcommand{\SRel}{\textit{SRel}\xspace} 
\newcommand{\ERel}{\textit{TRel}\xspace} 
\newcommand{\msn}[1]{\textnormal{\textit{#1}}}
\newcommand{\Type}{\msn{Type}\xspace}

\newcommand{\Exp}{\msn{Tm}\xspace}
\newcommand{\Val}{\msn{Val}\xspace}
\newcommand{\Ctxt}{\msn{Stk}\xspace}
\newcommand{\Stk}{\Ctxt}
\newcommand{\Subst}{\msn{Subst}\xspace}
\renewcommand{\r}{\ensuremath{R}}
\newcommand{\s}{\ensuremath{S}}
\newcommand{\f}{\ensuremath{F}}
\newcommand{\omegaOne}{\ensuremath{{\omega_1}}}

\newcommand{\fv}[1]{\ensuremath{\mathit{fv}(#1)}}
\newcommand{\ftv}[1]{\ensuremath{\mathit{ftv}(#1)}}
\newcommand{\IMPLIES}{\ensuremath{\Rightarrow}}
\newcommand{\IFF}{\ensuremath{\Leftrightarrow}}

\newcommand{\id}{\ensuremath{\mathit{id}}}

\newcommand{\nequiv}[1]{\ensuremath{\mathrel{\stackrel{#1}{=}}}}

\newcommand{\LATER}{\ensuremath{\mathop{\triangleright}}\xspace}

\newcommand{\syn}[1]{\textnormal{\texttt{#1}}}
\newcommand{\bnfeq}{\ensuremath{::=}}
\newcommand{\hmid}{\ensuremath{\;|\;}}
\newcommand{\ONE}{\ensuremath{\mathbf{1}}}
\newcommand{\NAT}{\syn{nat}}
\newcommand{\ALL}[2]{\ensuremath{\forall#1.#2}}
\newcommand{\REC}[2]{\ensuremath{\mu#1.#2}}
\newcommand{\IN}[1]{\ensuremath{\syn{in}_{#1}}}
\newcommand{\sto}{\ensuremath{\mathop{\multimap}}}
\newcommand{\unitval}{\langle\rangle}
\newcommand{\PAIR}[2]{\ensuremath{\langle#1,#2\rangle}}
\newcommand{\FUN}[2]{\ensuremath{\lambda#1.#2}}
\newcommand{\TFUN}[2]{\ensuremath{\Lambda#1.#2}}
\newcommand{\TAPP}[2]{\ensuremath{#1\,#2}}
\newcommand{\CASE}[1]{\ensuremath{\syn{case}~{#1}~\syn{of}}}
\newcommand{\BRANCH}[3]{\ensuremath{\IN{#1}\,#2.\,#3}}
\newcommand{\SEP}{\xspace\ensuremath{\texttt{|}}\xspace}
\newcommand{\IFZERO}[3]{\ensuremath{\syn{ifz}~#1~\syn{then}~#2~\syn{else}~#3}}
\newcommand{\CHOOSE}{\syn{?}\xspace}
\newcommand{\OR}{\ensuremath{\mathrel{\textnormal{\syn{or}}}}}
\newcommand{\FIX}{\syn{fix}\xspace}

\newcommand{\stepsto}{\ensuremath{\mathrel{\longmapsto}}}
\newcommand{\stepstopure}{\ensuremath{\mathrel{\stackrel{p}{\leadsto}}}}
\newcommand{\stepstozero}{\ensuremath{\mathrel{\leadsto^0}}}
\newcommand{\stepstozeropure}{\ensuremath{\mathrel{\stackrel{p}{\leadsto}^0}}}
\newcommand{\stepstoone}{\ensuremath{\mathrel{\leadsto^1}}}

\newcommand{\subst}[2]{\ensuremath{[{#2}{/}{#1}]}}
\newcommand{\emp}{\ensuremath{\varnothing}}  

\newcommand{\R}{\ensuremath{\mathcal{R}}\xspace} 

\newcommand{\ctxtapprox}{\ensuremath{\mathrel{\lesssim}^\textit{ctx}}}
\newcommand{\mayctxtapprox}{\ensuremath{\mathrel{\ctxtapprox_\downarrow}}}
\newcommand{\mustctxtapprox}{\ensuremath{\mathrel{\ctxtapprox_\Downarrow}}}

\newcommand{\ctxtequiv}{\ensuremath{\mathrel{\cong}^\textit{ctx}}}
\newcommand{\mayctxtequiv}{\ensuremath{\mathrel{\ctxtequiv_\downarrow}}}
\newcommand{\mustctxtequiv}{\ensuremath{\mathrel{\ctxtequiv_\Downarrow}}}
\newcommand{\ciu}{\ensuremath{\mathrel{\lesssim}^\textit{ciu}}}
\newcommand{\mayciu}{\ensuremath{\mathrel{\ciu_\downarrow}}}
\newcommand{\mustciu}{\ensuremath{\mathrel{\ciu_\Downarrow}}}

\newcommand{\logapprox}{\ensuremath{\mathrel{\lesssim}^\textit{log}}}

\newcommand{\maylogapprox}{\ensuremath{\mathrel{\logapprox_\downarrow}}}

\newcommand{\mustlogapprox}{\ensuremath{\mathrel{\logapprox_\Downarrow}}}

\newcommand{\LR}[1]{\ensuremath{\den{#1}}}
\newcommand{\TT}[1]{\ensuremath{{#1}^{\bot\!\!\bot}}}
\newcommand{\T}[1]{\ensuremath{{#1}^{\bot}}}

\newcommand{\den}[1]{%
 \left\llbracket #1
 \right\rrbracket}

\renewcommand{\paragraph}{\subsection}

\begin{document}

\title[Step-Indexed Relational Reasoning for Countable Nondeterminism]{Step-Indexed Relational Reasoning for Countable Nondeterminism\rsuper*}

\author[L.~Birkedal]{Lars Birkedal\rsuper a}	
\address{{\lsuper{a,b}}Aarhus University}	
\email{\{birkedal,abizjak\}@cs.au.dk}  

\author[A.~Bizjak]{Ale\v{s} Bizjak\rsuper b}	
\address{\vspace{-18 pt}}	

\author[J.~Schwinghammer]{Jan Schwinghammer\rsuper c}	
\address{{\lsuper c}Saarland University}	
\email{jan@ps.uni-saarland.de}  




\titlecomment{{\lsuper*}A preliminary version of this work has been presented at
 the 20th EACSL Annual Conference on Computer Science Logic (CSL'11),
 12-15 September 2011, Bergen, Norway, see
 Section~\ref{sec:comparison-conf}
 for how the present paper relates to the conference paper.}



\begin{abstract}
  \noindent Programming languages with countable nondeterministic choice are
 computationally interesting since countable nondeterminism arises
 when modeling fairness for concurrent systems.
 Because countable choice introduces non-continuous behaviour, it is 
 well-known that developing semantic models for programming languages 
 with countable nondeterminism is challenging. 
 We present a step-indexed logical relations model of a higher-order functional
 programming language with countable nondeterminism and demonstrate
 how it can be used to reason about contextually defined may- and
 must-equivalence.  
 In earlier step-indexed models, the indices have been drawn from
 $\omega$. Here the step-indexed relations for must-equivalence
 are indexed over an ordinal greater than $\omega$.
\end{abstract}


\maketitle


\section{Introduction}
\label{sec:intro}

Programming languages with countable nondeterministic choice are
computationally interesting since countable nondeterminism arises
when modeling fairness for concurrent systems. 
In this paper we show how to construct simple semantic models
for reasoning about may- and must-equivalence in a 
call-by-value higher-order functional programming language with countable nondeterminism,
recursive types and impredicative polymorphism.

Models for languages with   nondeterminism have originally
been studied using denotational techniques.  In the case of 
countably branching nondeterminism it is not enough to consider standard
$\omega$-continuous complete partial orders and the denotational models
become quite involved \cite{Apt:Plotkin:86,DiGianantonio:Honsell:Plotkin:95}. 
This has sparked research in operationally-based
theories of equivalence for nondeterministic higher-order languages 
\cite{Agha:Mason:Smith:Talcott:97,Lassen:PhD,Lassen:Moran:99,Lassen:Pitcher:97,Levy:06,Sabel:Schmidt-Schauss:08}.
In particular, Lassen investigated operationally-based
relational methods for countable nondeterminism and suggested that it would be interesting to
consider also methods based on logical
relations, i.e., where the \emph{types} of
the programming languages are given a relational interpretation~\cite[page~47]{Lassen:PhD}. Such an interpretation 
would allow one to relate terms of different types, as needed for
reasoning about parametricity properties of polymorphic types.  

For  
languages with recursive types, however, logical relations cannot be defined by
induction on types. In the case of deterministic languages, 
this problem has been addressed by the technique 
of syntactic minimal invariance~\cite{Birkedal:Harper:99}  (inspired by
domain theory~\cite{Pitts:96}).  The idea here is that one proves that
a syntactically definable fixed point on a recursive type is
contextually equivalent to the identity function, and then uses a
so-called unwinding theorem for syntactically definable fixed points
when showing the existence of the logical relations.  However, in the
presence of countable nondeterminism it is not clear how to define the
unwindings of the syntactic fixed point in the programming language.
Indeed, Lassen proved an unwinding theorem for his language with
countable nondeterminism, but he did so by extending the language with
new terms needed for representing the unwindings and left open the
question of whether this is a conservative extension of the language. 

Here we give a logical relations model of our language where we do not
rely on syntactic minimal invariance for constructing the logical
relations. Instead, we use the idea of step-indexed logical
relations~\cite{Appel:McAllester:01}.
In particular, we show how to use step-indexing over ordinals
larger than $\omega$ to reason about must-equivalence in the presence
of countable nondeterminism. 

This approach turns out to be both  {simple} and also  {useful} for reasoning
about concrete may- and must-equivalences.  We show that our logical
relations are sound and complete with
respect to the contextually defined notions of may- and
must-equivalence.  Moreover, we show how to use our logical relations
to establish some concrete equivalences. In particular, we prove the
recursion-induction rule from Lassen~\cite{Lassen:PhD} and establish
the syntactic minimal invariance property (without extending the
language with new unwinding terms). We also include an example to
show that the model can be used to prove parametricity properties
(free theorems) of polymorphic types.

\paragraph*{Overview of the technical development}

One way to understand the failure of $\omega$-continuity in an
operational setting is to consider the must-convergence predicate
$e\Downarrow$, which by Tarski's fixed point theorem can be defined as the least fixed point of the
monotone functional $\Phi(R) = \{e\;|\; \forall e'.\ e\stepsto e'\
\Rightarrow\ e'\in R\}$ on sets of terms. Here $e\stepsto e'$ means that $e$
reduces to $e'$ in one step. However, due to the countable branching the fixed point is not reached by
$\omega$-many iterations $\bigcup_{n\in\omega} \Phi^n(\emptyset)$. The reason is that even when a program has no infinite reduction sequences, we cannot in general bound the length of reduction sequences by any $n<\omega$.  

The idea of step-indexed semantics is a stratified construction of relations which facilitates the interpretation of recursive types, 
and in previous applications   this stratification has typically been realized by indexing over  $\omega$. 
However, as we pointed out, the closure ordinal of the inductively defined must-convergence predicate is strictly larger than $\omega$: 
the least fixed point $\Downarrow$ is reached after  $\omegaOne$-many iterations, for
$\omegaOne$   the least uncountable ordinal. (In fact, the least non-recursive ordinal would suffice \cite{Apt:Plotkin:86}.) 
Thus, one of the key steps in our development is the definition of $\alpha$-indexed uniform relations, for arbitrary ordinals $\alpha$, in Section~\ref{sect:uniform}. 

In Section~\ref{sect:maytheory} we define a logical $\omega$-indexed
uniform relation, and use this relation to prove a CIU theorem for
may-contextual equivalence. The logical relation combines
step-indexing and biorthogonality, and we can prove that it coincides
with may-contextual equivalence; the proofs are similar to those in
\cite{Pitts:10}.  Section~\ref{sect:musttheory} considers the case of
must-contextual equivalence.  The only modifications that this
requires, compared to Section~\ref{sect:maytheory}, are the use of
$\omegaOne$-indexed uniform relations and of a suitably adapted notion
of biorthogonality.


\paragraph*{Summary of contributions}
In summary, the contribution of this paper is a simple,
operationally-based logical relations model of countable nondeterminism in a
higher-order language, and the use of this model for proving several
non-trivial applications in Section~\ref{sect:applications}. In
particular, we derive a least prefixed point property for recursive
functions in our language and characterize the elements of the type
$\ALL{\alpha}{\alpha\times\alpha\to\alpha}$, using relational 
parametricity.

Laird \cite{Laird:06} has developed a fully abstract denotational 
model based on bidomains for a calculus similar to the one studied here
but without recursive and polymorphic types; 
our model appears to be the first model of countable nondeterminism 
for a language with impredicative polymorphism. Finite nondeterminism
and polymorphism has been studied for a call-by-name language 
by Johann et. al.~\cite{johann:2010}, who developed an
operational theory for algebraic effects.

\section{A lambda calculus with countable choice}
\label{sec:language}

\paragraph*{Syntax and operational semantics}

\begin{figure*}
\begin{align*}
\tau &\bnfeq \alpha\hmid\ONE\hmid\tau_1\times\tau_2\hmid\tau_1\to\tau_2\hmid\REC\alpha{\tau_1\mathop+\ldots\mathop+\tau_n}\hmid\ALL\alpha\tau\\
v  &\bnfeq x\hmid  \unitval\hmid \PAIR {v_1}{v_2}\hmid \FUN xe \hmid \IN{i}\,v
\hmid \TFUN\alpha e\\
e  &\bnfeq v \hmid 
\CHOOSE\hmid  \syn{proj}_i\,v\hmid v\,e\hmid \CASE{v}~\BRANCH{1}{x_1}{e_1}\SEP\ldots\SEP\BRANCH{n}{x_n}{e_n}\hmid
\TAPP v\tau\\
E  &\bnfeq []\hmid v\,E
\end{align*}
\caption{\label{fig:syntax} Types, terms and evaluation contexts}
\end{figure*}

Figure~\ref{fig:syntax} gives the syntax of a higher-order functional language with recursive and polymorphic types, and a (countably branching) choice construct. 
We assume disjoint, countably infinite sets of \emph{type variables}, ranged over by $\alpha$, and \emph{term variables}, ranged over by $x$. 
The free type variables of types and terms, $\ftv\tau$ and $\ftv e$, and free term variables $\fv e$, are defined in the usual way. 
The notation $(\cdot)\subst{\vec\alpha}{\vec\tau}$  denotes the simultaneous capture-avoiding substitution of types $\vec\tau$ for the free type variables $\vec\alpha$ in types and terms; similarly, $e\subst{\vec x}{\vec v}$ denotes simultaneous capture-avoiding substitution of values $\vec v$ for the free term variables $\vec x$ in $e$. 

\begin{figure*}
\begin{align*}
\syn{proj}_i\,\PAIR{v_1}{v_2} &\stepsto v_i
& 
\CASE{(\IN{j}\,v)}~(\ldots\SEP\BRANCH{j}{x_j}{e_j}\SEP\ldots) &\stepsto e_j\subst{x_j}{v}\\
(\FUN x e)\,v &\stepsto e \subst x v 
&
\CHOOSE &\stepsto \underline n \qquad\ \  (n\in\N)\\
\TAPP{(\TFUN\alpha e)}\tau &\stepsto e\subst \alpha\tau 
&
v\,e&\stepsto v\, e' \qquad\text{if $e\stepsto e'$}
\end{align*}
\caption{\label{fig:semantics} Operational semantics}
\end{figure*}

\newcommand{\letr}[3]{\ensuremath{\syn{let}~#1=#2~\syn{in}~#3}}

To reduce the number of proof cases in the formal
development, we keep the syntax minimal. For instance, we only include
$\syn{proj}_1 v$, for $v$ a value and not an expression. In examples we may use additional
syntactic sugar. We write
$\syn{let}~x=e~\syn{in}~e'$ for $(\FUN x{e'})\,e$ and $\TAPP e\tau$ for $\syn{let}~f=e~\syn{in}~\TAPP f\tau$ for some fresh $f$. 

We define the unary natural numbers datatype as 
$\NAT = \REC\alpha{\ONE\mathop+\alpha}$ and write 
$\underline 0 = \IN{1}\,\unitval$ and $\underline{n{+}1} = \IN{2}(\underline n)$.
The `erratic' (finitely branching) choice construct $e_1\OR{e_2}$ can be defined from \CHOOSE as $\syn{let}~x=\CHOOSE~\syn{in}~\CASE{x}~\BRANCH{1}{y}{e_1}~\SEP~\BRANCH{2}{y}{e_2}$ for fresh $x,y$.

The operational semantics of the language is given in Figure~\ref{fig:semantics} by a reduction relation $e\stepsto e'$. 
In particular, the choice operator \CHOOSE evaluates nondeterministically to any numeral $\underline n$ ($n\in\N$). 
We also consider evaluation contexts $E$, and write $E[e]$ for the term obtained by plugging $e$ into $E$. It is easy to see that $e\stepsto e'$ holds if and only if $E[e]\stepsto E[e']$. 

\begin{figure*}
\begin{gather*}
\inferrule{x{:}\tau\in\Gamma\\ \Delta\vdash\Gamma}{\Delta;\Gamma\vdash x:\tau}
\qquad
\inferrule{\Delta\vdash\Gamma }{\Delta;\Gamma\vdash\unitval:\ONE}
\qquad
\inferrule{\Delta;\Gamma\vdash v_1:\tau_1\\\Delta;\Gamma\vdash v_2:\tau_2 }{\Delta;\Gamma\vdash\PAIR{v_1}{v_2}:\tau_1\mathop\times\tau_2}
\\[2mm]
\inferrule{\Delta;\Gamma,x{:}\tau_1\vdash e:\tau_2\ }{\Delta;\Gamma\vdash\FUN xe:\tau_1\mathop\to\tau_2}
\qquad
\inferrule{\Delta;\Gamma\vdash v:\tau_j\subst{\alpha}{\REC\alpha{\tau_1\mathop+\ldots\mathop+\tau_n}}}{\Delta;\Gamma\vdash\IN{j}\,v:\REC\alpha{\tau_1\mathop+\ldots\mathop+\tau_n}}\ 1\leq j\leq n
\\[2mm]
\inferrule{\Delta,\alpha;\Gamma\vdash e:\tau}{\Delta;\Gamma\vdash\TFUN \alpha e : \ALL\alpha\tau}
\qquad
\inferrule{\Delta;\Gamma\vdash v:\tau_1\times\tau_2}{\Delta;\Gamma\vdash \syn{proj}_i\,v:\tau_i}
\qquad
\inferrule{\Delta;\Gamma\vdash v:\tau'\to\tau\\ \Delta;\Gamma\vdash e:\tau'}{\Delta;\Gamma\vdash v\,e:\tau}
\\[2mm]
\inferrule{\Delta;\Gamma\vdash v:\REC\alpha{\tau_1\mathop+\ldots\mathop+\tau_n}\\
\ldots\\ \Delta;\Gamma,x_j{:}\tau_j\subst{\alpha}{\REC\alpha{\tau_1\mathop+\ldots\mathop+\tau_n}}\vdash e_j:\tau \\ \ldots}
{\Delta;\Gamma\vdash \CASE{v}\,(\ldots\SEP\,\BRANCH{j}{x_j}{e_j}\,\SEP\ldots) : \tau}
\\[2mm]
\inferrule{\Delta;\Gamma\vdash v:\ALL\alpha\tau\\\Delta\vdash\tau'}{\Delta;\Gamma\vdash\TAPP v{\tau'}:\tau\subst\alpha{\tau'}}
\qquad
\inferrule{ \Delta\vdash\Gamma}{\Delta;\Gamma\vdash \CHOOSE :\NAT}
\\[4mm]
\inferrule{\emp\vdash\tau }{  \vdash [] : \tau\sto \tau}
\qquad
\inferrule{\emp;\emp\vdash v:\tau\to\tau_2\\  \vdash E: \tau_1\sto\tau}{  \vdash v\,E: \tau_1\sto \tau_2}
\end{gather*}
\caption{\label{fig:typing} 
Typing of terms and evaluation contexts, where $\Gamma\bnfeq\emptyset\hmid\Gamma,x{:}\tau$ and $\Delta\bnfeq\emptyset\hmid\Delta,\alpha$. 
The notation $\Delta\vdash\tau$ means that $\ftv{\tau}\subseteq\Delta$, and $\Delta\vdash\Gamma$ means that $\Delta\vdash\tau$ holds for all $x{:}\tau\in\Gamma$. 
}
\end{figure*}

Typing judgements take the form $\Delta;\Gamma\vdash e:\tau$ where $\Gamma$ is a typing context $x_1{:}\tau_1,\ldots,x_n{:}\tau_n$ and where $\Delta$ is a finite set of type variables that contains the free type variables of $\tau_1,\ldots,\tau_n$ and $\tau$.   
The rules defining this judgement are  summarized in Figure~\ref{fig:typing}. 
The typing judgement for evaluation contexts, $\vdash E:\tau\sto\tau'$, 
means that $\emp;\emp\vdash E[e]:\tau'$ holds whenever $\emp;\emp\vdash e:\tau$. 

We write $\Type$ for the set of closed types $\tau$, i.e., where $\ftv\tau=\emptyset$. 
We write $\Val(\tau)$ and $\Exp(\tau)$ for the sets of closed values and terms of type $\tau$, resp., and  
$\Ctxt(\tau)$ for the set of $\tau$-accepting evaluation contexts. 
For a typing context $\Gamma = x_1{:}\tau_1,\ldots,x_n{:}\tau_n$ with $\tau_1,\ldots,\tau_n\in\Type$, let 
$\Subst(\Gamma) = \{\gamma\in\Val^{\vec x}\;|\;\forall 1\leq i\leq n.\  \gamma(x_i)\in\Val(\tau_i) \}$ 
denote the set of type-respecting value substitutions. 
In particular, if  $\Delta;\Gamma\vdash e:\tau$ then $\emp;\emp\vdash
e\delta \gamma : \tau\delta$ for any $\delta\in \Type^\Delta$ and
$\gamma\in\Subst(\Gamma\delta)$, and the type system satisfies 
standard properties:

\begin{lem}
[Canonical forms]\ 
\begin{itemize}
\item If $v\in\Val{(\ONE)}$ then $v$ is $\unitval$.
\item If $v\in\Val{(\tau_1\mathop\times\tau_2)}$ then $v$ is of the form  $\PAIR{v_1}{v_2}$ for $v_i\in\Val(\tau_i)$. 
\item If $v\in\Val{(\tau_1\mathop\to\tau_2)}$ then $v$ is of the form $\FUN{x}{t}$ for some $x$ and $e$. 
\item If $v\in\Val(\REC\alpha{\tau_1\mathop+\ldots\mathop+\tau_m})$ then $v$ is of the form $\IN{j}\,v'$ for some $1\leq j\leq m$ and $v'\in\Val(\tau_j\subst\alpha{\REC\alpha{\tau_1\mathop+\ldots\mathop+\tau_m}})$.  
\item If $v\in\Val(\ALL\alpha\tau)$ then $v$ is of the form $\TFUN\alpha e$ for some $\alpha$ and $e$. 
\end{itemize}
\end{lem}

\begin{prop}[Preservation and progress]\ 
\begin{itemize}
\item If $e\in\Exp(\tau)$ and $e\stepsto e'$ then $e'\in\Exp(\tau)$. 
\item If $e\in\Exp(\tau)\setminus\Val(\tau)$ then $e\stepsto e'$ for some $e'$. 
\end{itemize}
\end{prop}

Following Lassen~\cite{Lassen:PhD}, 
we let $\FIX :
\ALL{\alpha,\beta}{((\alpha\mathop\to\beta)\to(\alpha\mathop\to\beta))\to(\alpha\mathop\to\beta)}$
denote a variant of the (call-by-value) fixed point combinator from
untyped lambda calculus, $\FIX =
\TFUN{\alpha,\beta}{\FUN{f}{\delta_f(\IN{}\,\delta_f)}}$ where
$\delta_f$ is the term
$\FUN{y}{\CASE{y}~\BRANCH{}{y'}{f(\FUN{x}{\syn{let}~r\mathop=y'\,y~\syn{in}~r\,x})}}$.
In Section~\ref{sect:applications} we show in what sense $\FIX$ is a
fixed point combinator. We write $\Omega :\ALL\alpha\alpha$ for the term
$\TFUN\alpha{\TAPP{\FIX}{\ONE\,\alpha}\,(\FUN{f}{f})\,\unitval}$.
Note that, for all closed types $\tau$, reduction from
$\TAPP{\Omega}{\tau}$ is deterministic and non-terminating.

\paragraph*{Contextual approximation}

We follow Lassen's approach   \cite{Lassen:PhD} and define contextual approximation 
as the largest relation that satisfies certain compatibility and adequacy properties (also see, e.g.\ \cite{Pitts:05,Pitts:10}). 
The technical advantage of this approach, compared to the more traditional one of universally quantifying over program contexts, is that in proofs there will be no need to explicitly take care of contexts and of term occurrences within contexts. 
In our terminology, we keep close to Pitts \cite{Pitts:05}, except for suitably adapting the definitions to take the nondeterministic outcomes of evaluation into account. 

The observables on which contextual approximation is based are given
by may- and must-convergence.  A closed term $e$ \emph{may-converges},
written $e\downarrow$, if $e\stepsto^* v$ for some $v\in\Val$, and $e$
\emph{may-diverges}, written $e\uparrow$, if there is an infinite
reduction sequence starting from $e$.  The \emph{must-convergence}
predicate $e\Downarrow$ is the complement of may-divergence, and it
can be defined as the least predicate satisfying
$e\Downarrow$ if for all $e'$, if $e\stepsto e'$ then
$e'\Downarrow$. In addition, we say that $e$ \emph{must-diverges} if it does not may-converge.

\begin{defi}
[Type-indexed relation]
A \emph{type-indexed relation} is a set consisting of tuples $(\Delta,\Gamma,e,e',\tau)$ 
such that $\Delta;\Gamma\vdash e:\tau$ and $\Delta;\Gamma\vdash e':\tau$ holds, where
we write $\Delta;\Gamma\vdash e\mathrel{\R} e':\tau$ for  $(\Delta,\Gamma,e,e',\tau)\in\R$. 
\end{defi}

\begin{defi}
[Precongruence]
A type-indexed relation \R is \emph{reflexive} if 
$\Delta;\Gamma\vdash e:\tau$ implies $\Delta;\Gamma\vdash e\mathrel\R e:\tau$. 
It is \emph{transitive} if 
$\Delta;\Gamma\vdash e\mathrel\R e':\tau$ and $\Delta;\Gamma\vdash e'\mathrel\R e'':\tau$ implies $\Delta;\Gamma\vdash e\mathrel\R e'':\tau$. 
A \emph{precongruence} is a reflexive and transitive type-indexed relation \R that is closed under the inference rules in 
Figure~\ref{fig:compatibility}. 
\end{defi}

\begin{figure}[t]
\begin{gather*}
\inferrule{ }{\Delta;\Gamma\vdash x\mathrel\R x:\tau}\ x{:}\tau\in\Gamma
\qquad
\inferrule{  }{\Delta;\Gamma\vdash\unitval \mathrel\R \unitval:\ONE}
\\[2mm]
\inferrule{\Delta;\Gamma\vdash v_1\mathrel\R v_1':\tau_1\\\Delta;\Gamma\vdash v_2\mathrel\R v_2':\tau_2 }{\Delta;\Gamma\vdash\PAIR{v_1}{v_2}\mathrel\R \PAIR{v_1'}{v_2'}:\tau_1\times\tau_2}
\qquad
\inferrule{\Delta;\Gamma,x{:}\tau_1\vdash e\mathrel\R e':\tau_2 }{\Delta;\Gamma\vdash\FUN xe\mathrel\R\FUN x{e'}:\tau_1\to\tau_2}
\\[2mm]
\inferrule{\Delta;\Gamma\vdash v\mathrel\R v':\tau_j\subst\alpha{\REC\alpha{\tau_1\mathop+\ldots\mathop+\tau_n}}}{\Delta;\Gamma\vdash\IN{j}\,v\mathrel\R \IN{j}\,v':\REC\alpha{\tau_1\mathop+\ldots\mathop+\tau_n}}\ 1\leq j\leq n
\qquad
\inferrule{\Delta,\alpha;\Gamma\vdash e\mathrel\R e':\tau}{\Delta;\Gamma\vdash\TFUN \alpha e\mathrel\R\TFUN \alpha {e'} : \ALL\alpha\tau}
\\[2mm]
\inferrule{\Delta;\Gamma\vdash v\mathrel\R v':\tau_1\times\tau_2}{\Delta;\Gamma\vdash \syn{proj}_i\,v\mathrel\R\syn{proj}_i\,v':\tau_i}
\qquad
\inferrule{\Delta;\Gamma\vdash v\mathrel\R v':\tau'\to\tau\\ \Delta;\Gamma\vdash e\mathrel\R e':\tau'}{\Delta;\Gamma\vdash v\,e\mathrel\R v'\,e':\tau}
\\[2mm]
\inferrule{\Delta;\Gamma\vdash v\mathrel\R v':\tau\\
\ldots\\
\Delta;\Gamma,x_j{:} \tau_j\subst{\alpha}{\tau}\vdash e_j\mathrel\R e_j':\tau'\\
\ldots\\
}
{\Delta;\Gamma\vdash \CASE{v}\,(\ldots\SEP\,\BRANCH{j}{x_j}{e_j}\,\SEP\ldots) \mathrel\R\CASE{v'}\,(\ldots\SEP\,\BRANCH{j}{x_j}{e_j}\,\SEP\ldots): \tau'}\ \tau=\REC\alpha{\tau_1\mathop+\ldots\mathop+\tau_n}
\\[2mm]
\inferrule{\Delta;\Gamma\vdash v\mathrel\R v':\ALL\alpha\tau}{\Delta;\Gamma\vdash\TAPP v{\tau'}\mathrel\R\TAPP{v'}{\tau'}:\tau\subst\alpha{\tau'}}\ \ftv{\tau'}\subseteq\Delta
\qquad
\inferrule{  }{\Delta;\Gamma\vdash \CHOOSE\mathrel\R\CHOOSE :\NAT}
\end{gather*}
\caption{\label{fig:compatibility} Compatibility properties of type-indexed relations}
\end{figure}

\begin{defi}
[May- and must-adequate relations] 
A type-indexed relation \R is \emph{may-adequate} if, whenever  
$\emp;\emp\vdash e\mathrel{\R} e':\tau$ holds, then $e\downarrow$ implies $e'\downarrow$. 
It is \emph{must-adequate} if, whenever $\emp; \emp\vdash e\mathrel{\R} e':\tau$ holds, then $e\Downarrow$ implies $e'\Downarrow$. 
\end{defi}

\begin{defi}
[Contextual approximations and equivalences]
\emph{May-contextual approximation}, written $\mayctxtapprox$, is the largest may-adequate precongruence. 
\emph{May-contextual equivalence}, $\mayctxtequiv$, is the symmetrization of  $\mayctxtapprox$. 
Analogously, \emph{must-contextual approximation}, written $\mustctxtapprox$, is the largest must-adequate precongruence, and 
\emph{must-contextual equivalence}, $\mustctxtequiv$, is its symmetrization.   
\emph{Contextual approximation}, $\ctxtapprox$, and \emph{contextual equivalence}, $\ctxtequiv$, are given as intersections of the respective may- and must-relations,  
and thus ${\ctxtequiv}$ is also the symmetrization of $\ctxtapprox$. 
\end{defi}

That this largest (may-, must-) adequate precongruence exists can be shown as in \cite{Pitts:05}, by proving that  
the relation $S=\bigcup\{R\;|\;\text{$R$ compatible and (may-, must-) adequate}\}$  is   an adequate precongruence. 


In principle, to establish an equivalence $\Delta;\Gamma\vdash
e\ctxtequiv e':\tau$ it suffices to find some may- and must-adequate congruence
$\R$ that contains the tuple $(\Delta,\Gamma,e,e',\tau)$ since
$\ctxtequiv$ is the largest such relation.  However, in practice it is
difficult to verify that a relation $\R$ has the necessary compatibility
properties   in Figure~\ref{fig:compatibility}.  An alternative
characterization of the contextual approximation and equivalence
relations can be  given in terms of  CIU preorders \cite{Mason:Talcott:91}, which we
define next.

\begin{defi}
[CIU preorders] 
\emph{May- and must-CIU preorder}, written $\mayciu$ and $\mustciu$ resp., are the type-indexed relations defined as follows:
for all $e,e'$ with $\Delta;\Gamma\vdash e:\tau$ and $\Delta;\Gamma\vdash e':\tau$, 
\begin{itemize}
\item $\Delta;\Gamma\vdash e\mayciu e':\tau\ 
\IFF\ 
\forall\delta \mathop\in\Type^\Delta,\,  \gamma\mathop\in\Subst(\Gamma\delta),\, E \mathop\in\Ctxt(\tau\delta).\ 
{E[e \delta\gamma]\downarrow}\, \Rightarrow\, E[e' \delta\gamma]\downarrow$
\item $\Delta;\Gamma\vdash e\mustciu e':\tau\ 
\IFF\ 
\forall\delta \mathop\in\Type^\Delta,\,  \gamma \mathop\in\Subst(\Gamma\delta),\, E \mathop\in\Ctxt(\tau\delta).\ 
{E[e \delta \gamma]\Downarrow}\, \Rightarrow\, E[e' \delta\gamma]\Downarrow$
\end{itemize}
The CIU preorder is defined as the intersection of $\mayciu$ and $\mustciu$. 
\end{defi}

\begin{thm}
[CIU theorem] 
\label{thm:ciu}
The (may-, must-) CIU preorder coincides with (may-, must-) contextual approximation.
\end{thm}

Using the CIU theorem, it is easy to verify that all the deterministic reductions are also valid equivalences, and that the various call-by-value  eta laws hold. 
Moreover, we can establish the laws of Moggi's computational lambda calculus and the basic (inequational) theory of erratic choice (Figure~\ref{fig:basictheory}). 
We will prove the CIU theorem in Section~\ref{sect:maytheory} (for the may-CIU preorder) and Section~\ref{sect:musttheory} (for the must-CIU preorder). 
The CIU theorem was also proved, using different operational
techniques, for a langauge with countable nondeterminism (but no
polymorphism) in~\cite{Lassen:PhD}. 

\begin{figure}
\begin{align*}
\syn{let}~x\mathop=\CHOOSE~\syn{in}~e &\ctxtequiv e\ \ (x\notin\fv{e})
&
\syn{let}~x\mathop=v~\syn{in}~e &\ctxtequiv e\subst x v
&
\syn{let}~x\mathop=e~\syn{in}~x &\ctxtequiv e
\\
e \OR e &\ctxtequiv e
&
\TAPP{\Omega}{\tau}&\mayctxtapprox e
&
\TAPP{\Omega}{\tau} &\mustctxtapprox e
\\
e_1\OR e_2 &\ctxtequiv e_2\OR e_1
&
e_1 &\mayctxtapprox  e_1\OR e_2 
&
e_1\OR e_2 &\mustctxtapprox e_1
\\
(e_1\OR e_2)\OR e_3 &\ctxtequiv e_1\OR (e_2\OR e_3) 
&
e\OR(\TAPP{\Omega}{\tau}) &\mayctxtequiv e
&
e\OR(\TAPP{\Omega}{\tau}) &\mustctxtequiv \TAPP{\Omega}{\tau}
\end{align*}
\caption{\label{fig:basictheory} Basic may- and must-theory, where
  $e_1\OR{e_2}$ is an abbreviation for the term
  $\syn{let}~x=\CHOOSE~\syn{in}~\CASE{x}~\BRANCH{1}{y}{e_1}~\SEP~\BRANCH{2}{y}{e_2}$,
and $e$ is of type $\tau$.}
\end{figure}

\section{Uniform relations} 
\label{sect:uniform}
For a limit ordinal number $\alpha$ and a set $X$ we define an
\emph{$\alpha$-indexed uniform subset on $X$} to be a family
$(\r_\beta)_{\beta<\alpha}$ of subsets $\r_\beta\subseteq X$ such
that
\begin{itemize}
\item $\r_0 = X$, 
\item $\r_{\beta+1}\subseteq \r_{\beta}$ for all $\beta<\alpha$, and 
\item $\r_\lambda = \bigcap_{\beta<\lambda} \r_\beta$ for every limit ordinal $\lambda<\alpha$.
\end{itemize}
Let $\ordURel\alpha(X)$ denote the $\alpha$-indexed uniform subsets on $X$. 

\paragraph*{Recursive definitions}

The notions of $n$-equivalence, non-expansiveness and contractiveness
(e.g., \cite{Birkedal:Reus:Schwinghammer:Stovring:Thamsborg:Yang:11})
all generalize from the case of $\omega$-indexed uniform subsets:
Given $\alpha$-indexed uniform subsets $\r,\s\in\ordURel\alpha(X)$
and $\nu<\alpha$ we say that $\r$ and $\s$ are
\emph{$\nu$-equivalent}, written $\r\nequiv\nu \s$, if
$\r_\beta=\s_\beta$ for all $\beta\leq\nu$.  In particular, $\r= \s$
if and only if $\r\nequiv\nu \s$ for all $\nu<\alpha$.  

A function $\f: \ordURel\alpha(X_1)\mathop\times\cdots\mathop\times\ordURel\alpha(X_n)\to\ordURel\alpha(X)$ is \emph{non-expansive} if
$\vec\r\nequiv\nu\vec\s$ implies ${\f(\vec \r)}\nequiv\nu {\f(\vec\s)}$, and $\f$ is
\emph{contractive} if $\vec\r\nequiv\nu \vec\s$ implies
${\f(\vec\r)}\nequiv{\nu\mathop+1} {\f(\vec\s)}$. 
If $\r\in\ordURel\alpha(X)$ then $\LATER \r\in\ordURel\alpha(X)$ is
the uniform subset determined by $({\LATER \r})_{\beta\mathop+1} =
\r_\beta$; this operation gives rise to a contractive function on
$\ordURel\alpha(X)$. Henceforth, we often omit parentheses
and write ${\LATER \r}_\beta$ for $({\LATER \r})_\beta$.

\begin{prop}
[Unique fixed points]
\label{prop:unique-fixed-points}
If $\f : \ordURel\alpha(X)\to\ordURel\alpha(X)$ is  contractive, 
then $\f$ has a unique fixed point $\textit{fix}\,r.\f(r)$. 
\end{prop}
\begin{proof} 
First note that $\f$ has at most one fixed point: if $\r,\s$ are fixed points of $\f$ then, by the contractiveness of $\f$, we can establish that  $\r = F(\r)\nequiv\nu \f(\s) = \s$ holds for all $\nu<\alpha$ by induction and thus $\r=\s$. 

Because of the uniformity conditions it is sufficient to  give  the components of  the fixed point  $\textit{fix}\,r.\f(r)$  
 that are indexed by successor ordinals. 
We set $\textit{fix}\,r.\f(r)_{\nu\mathop+1} = \f(\r)_{\nu\mathop+1}$ where 
$\r\in\ordURel\alpha(X)$ is defined by 
$\r_\beta = \textit{fix}\,r.\f(r)_\beta$ for $\beta\leq\nu$ and $\r_\beta = \emptyset$ for $\beta>\nu$.  
By induction, it is easy to see that $\textit{fix}\,r.\f(r)\in\ordURel\alpha(X)$ 
and that $\f(\textit{fix}\,r.\f(r))_\nu = \textit{fix}\,r.\f(r)_\nu$ holds for all $\nu<\alpha$, 
and thus $\f(\textit{fix}\,r.\f(r)) = \textit{fix}\,r.\f(r)$. 
\end{proof}

Proposition~\ref{prop:unique-fixed-points} is an instance of Di
Gianantonio and Miculan's     sheaf-theoretic fixed point
theorem \cite{DiGianantonio:Miculan:04}. Indeed, an $\alpha$-indexed
uniform subset on $X$ corresponds to a subobject of the constant
sheaf on $X$ in the sheaf topos on $\alpha$. 

\paragraph*{Uniform relations on syntax}

For $\tau,\tau'\in\Type$ we consider the collections of
$\beta$-indexed \emph{uniform relations}  between values, terms and
evaluation contexts: 
we write $\VRel_\beta(\tau,\tau')$ for $\ordURel\beta(\Val(\tau)\mathop\times\Val(\tau'))$,
we write $\SRel_\beta(\tau,\tau')$ for $\ordURel \beta(\Ctxt(\tau)
\mathop\times\Ctxt(\tau'))$, and we use $\ERel_\beta(\tau,\tau')$ for $\ordURel
\beta(\Exp(\tau) \mathop\times\Exp(\tau'))$.
Note that a value relation may relate values of distinct types; that
is essential for reasoning about relational parametricity, see, e.g., the
proof of Lemma~\ref{lem:must:ex:converging}.

The description of the logical relations in the sections below makes use of the following (non-expansive) constructions on uniform relations: \begin{itemize}
\item  
$\r_1\mathop\times \r_2 \in\VRel_\beta(\tau_1\mathop\times\tau_2,\tau_1'\mathop\times\tau_2')$, 
for $\r_1\in\VRel_\beta(\tau_1,\tau_1')$ and $\r_2\in\VRel_\beta(\tau_2,\tau_2')$,  is defined by 
$(\r_1\mathop\times \r_2)_\nu\, =\, \{(\PAIR{v_1}{v_2}, \PAIR{v_1'}{v_2'})\;|\;(v_1,v_1')\in (\r_1)_{\nu}\ \wedge\  (v_2,v_2')\in (\r_2)_{\nu}\}$.
\vspace{2mm}
\item 
$\r_1\mathop\to \r_2 \in\VRel_\beta(\tau_1\mathop\to\tau_2,\tau_1'\mathop\to\tau_2')$, for $\r_1\in\VRel_\beta(\tau_1,\tau_1')$ and $\r_2\in\ERel_\beta(\tau_2,\tau_2')$, is given by 
$(\r_1\mathop\to \r_2)_\nu = \{(\FUN xe,\FUN x{e'})\;|\;\forall \nu'\,{\leq}\, \nu.\, \forall (v,v')\,{\in}\, (\r_1)_{\nu'}.\, (e\subst{x}{v}, e'\subst{x}{v'})\,{\in}\,(\r_2)_{\nu'}\}$.%
\vspace{2mm}
%
\item $\forall
 r.F(r)\,{\in}\,\VRel_\beta(\ALL\alpha{\tau_1},\ALL\alpha{\tau_1'})$,
 for
 $F_{\tau,\tau'}:\VRel_\beta(\tau,\tau')\to\ERel_\beta(\tau_1\subst\alpha\tau,\tau_1'\subst\alpha{\tau'})$
 a family of non-expansive maps, is the uniform relation that is
 defined by \[\forall r.F(r)_\nu= \{(\TFUN\alpha
 e,\TFUN\alpha{e'})\;|\;\forall\tau,\tau'\,{\in}\,\Type,
 \r\,{\in}\,\VRel_\nu(\tau,\tau').\
 (e\subst\alpha\tau,e'\subst\alpha{\tau'}\in{F_{\tau,\tau'}(\r)}_\nu
 \}\ .\]
\item 
$\IN{j} \r\in\VRel_\beta(\tau,\tau')$, for $\tau=\REC\alpha{\tau_1\mathop{+}\ldots\mathop{+}\tau_m}$ and $\tau'=\REC\alpha{\tau_1'\mathop{+}\ldots\mathop{+}\tau_n'}$ as well as $\r\in\VRel_\beta(\tau_j\subst\alpha{\tau},\tau_j'\subst\alpha{\tau'})$, is given by 
$(\IN{j} \r)_\nu\, =\, \{(\IN{j}\,v,\IN{j}\,v')\;|\;(v,v')\in R_\nu \}$. 
\end{itemize}

\section{May equational theory}
\label{sect:maytheory}

In this section, we will define a logical uniform relation that is
used to prove that may-CIU preorder and may-contextual approximation
coincide.  The key idea of the definition is the usual one of
step-indexing \cite{Appel:McAllester:01}, i.e., that the observables
can be stratified based on step-counting in the operational semantics.
Let us refer to reduction steps of the form
\begin{displaymath}
\CASE{(\IN{j}\,v)}~(\ldots\SEP\BRANCH{j}{x_j}{e_j}\SEP\ldots) 
\stepsto e_j\subst{x_j}{v}
\end{displaymath}
as \emph{unfold-fold reductions}.
Following~\cite{BirkedalL:lslr-journal} we will only count such unfold-fold
reductions. The advantage of this is that 
the interpretation of types is slightly more extensional than if we 
counted all reduction steps; see the precise formulation in
Lemma~\ref{lem:ext-lemma} below.
Hence we define
\begin{displaymath}
e\stepstozero e'
\end{displaymath}
to mean that $e\stepsto^* e'$ and \emph{none} of the reductions in the
reduction sequence is an
unfold-fold reduction, and we define
\begin{displaymath}
e\stepstoone e'
\end{displaymath}
to mean that $e\stepsto^* e'$ and \emph{exactly one} of the reductions
in the reduction sequence is an
unfold-fold reduction.

We shall also make use of pure reductions.  To that end, we refer to
reductions of the form
\begin{displaymath}
 \CHOOSE \stepsto \underline n  
\end{displaymath}
as \emph{choice reductions}.
We then define
\begin{displaymath}
e\stepstopure e'
\end{displaymath}
to mean that $e \stepsto^* e'$ and \emph{none} of the reductions in
the reduction sequence is a choice reduction.
Further, we define
\begin{displaymath}
e\stepstozeropure e'
\end{displaymath}
to mean that $e\stepstozero e'$ and $e\stepstopure e'$.

We write $e\downarrow_n$ if $e \stepsto^* v$ for some
$v\in\Val$ and \emph{at most} $n$ reduction steps are unfold-fold reductions.

\paragraph*{Logical $\omega$-indexed uniform relation for may-approximation}

In the case of may-approximation, it suffices to consider $\omega$-indexed uniform relations. 
Using the   constructions on relations given above, 
we  define  a relational interpretation $\den{\tau}({\vec
 r})\in\mkern -4 mu\VRel_\omega(\tau\subst{\vec\alpha}{\vec\tau},\tau\subst{\vec\alpha}{\vec\tau'})$
by induction on the type $\vec\alpha\vdash\tau$, given 
closed types $\tau_1,\tau_1',\ldots,\tau_k,\tau_k'\in\Type$ and
relations
$r_1\in\VRel_\omega(\tau_1,\tau_1'),\ldots,r_k\in\VRel_\omega(\tau_k,\tau_k')$: 
\begin{align*}
\LR{\alpha_i}(\vec r) &= r_i &
\LR{\tau_1\times\tau_2}(\vec r) &= \LR{\tau_1}(\vec r)\mathop\times\den{\tau_2}(\vec r)\\
\LR{\ONE}(\vec r) &= (\textit{Id}_\ONE)_{n<\omega}&
\LR{\tau_1\to\tau_2}(\vec r) &=  \LR{\tau_1}(\vec r)\mathop\to\TT{(\LR{\tau_2}(\vec r))}\\
\LR{\ALL\alpha{\tau}}(\vec r) &= \forall r.\TT{(\den{\tau}(\vec r,r))} &
\den{\REC\alpha{\tau_1\mathop+\ldots\mathop+\tau_m}}(\vec r)&= \textit{fix}\, s.\,\textstyle{\bigcup_j} \IN{j}(\LATER \den{\tau_j}(\vec r,s))
\end{align*}
Here, value relations $r\in\VRel_\omega(\tau,\tau')$ are
lifted to relations $\T r\in\SRel_\omega(\tau,\tau')$ on evaluation contexts and to relations
$\TT r\in\ERel_\omega(\tau,\tau')$ on terms by biorthogonality, much as in~\cite{BirkedalL:stslr-conf}: 
\begin{align*}
\T{r}_n &= \{(E,E')\;|\;\forall j\leq n.\ \forall (v,v')\in r_j.\ E[v]\downarrow_j\ \Rightarrow\ E'[v']\downarrow\ \}\\
\TT{r}_n &= \{(e,e')\;|\;\forall j\leq n.\ \forall (E,E')\in\T{r}_j.\ E[e]\downarrow_j\ \Rightarrow\ E'[e']\downarrow\ \}
\end{align*}
The fixed point in the interpretation of recursive types is well-defined by Proposition~\ref{prop:unique-fixed-points} since each $\den{\tau}$ denotes a family of non-expansive functions, and thus composition with $\LATER$ yields a contractive function.   
Intuitively, we want to relate two values $\IN{1}{v}$ and $\IN{1}{v'}$ 
of a recursive type if $v$ and $v'$ are related at the unfolded
type. We cannot define the relation that way. Instead we only require
that $v$ and $v'$ are related at one step later. This 
suffices because we count unfold-fold reductions, see the
proof of Proposition~\ref{prop:FTLR} for details.

We often omit parentheses and write ${\den{\tau}\vec r}_n$ for
$({\den{\tau}\vec r})_n$ and
$\T{\den{\tau}\vec r}_n$ for
$\T{(\den{\tau}\vec r)}_n$ and
$\TT{\den{\tau}\vec r}_n$ for
$\TT{(\den{\tau}\vec r)}_n$.

The following lemmas express basic properties of the defined relations which are often used
in subsequent proofs and calculations.

\begin{lem}
[Substitution] 
\label{lem:may:subst}
If $\Delta,\alpha\vdash\tau$ and $\Delta\vdash\tau'$ 
then $\den{\tau\subst\alpha{\tau'}}(\vec r) = \den{\tau}(\vec r,\den{\tau'}(\vec r))$. 
\end{lem}

\begin{lem}
[Extensiveness]
\label{lem:extensiveness}
For all $r\in\VRel(\tau,\tau')$, 
$r\subseteq\TT{r}$.
\end{lem}

\begin{lem}
[Monotonicity]
For all $r,s\in\VRel(\tau,\tau')$, 
if $r\subseteq s$ then $\TT{r}\subseteq\TT{s}$.
\end{lem}

\begin{lem}
[Context composition]
\label{lem:context-composition}
If $(v,v')\in{\den{\tau_1\mathop\to\tau_2}\vec r}_n$ 
and  $(E,E')\in\T{\den{\tau_2}\vec r}_n$ 
then $(E[v \,[]],E'[v' \,[]])\in\T{\den{\tau_1}\vec r}_{n}$. 
\end{lem}
\begin{proof}
Let $j\leq n$, $(v_1,v_1')\in{\den{\tau_1}\vec r}_j$. 
Assume $E[v\,v_1]\downarrow_j$. 
We have $v=\FUN{x}{e}$ and $v'=\FUN{x}{e'}$ and
$(\FUN{x}{e},\FUN{x}{e'})\in{\den{\tau_1\mathop\to\tau_2}\vec r}_n$
for some $x,e,e'$ 
and since $E[v\,v_1]\stepsto E[e\subst{x}{v_1}]$ also
$E[e\subst{x}{v_1}]\downarrow_j$.
By definition, $(e\subst{x}{v_1},e'\subst{x}{v_1'})\in\TT{\den{\tau_2}\vec r}_j$. 
From $(E,E')\in\T{\den{\tau_2}\vec r}_n$  we obtain $E'[e'\subst{x}{v_1'}]\downarrow$. 
Thus, $E'[v'\,v_1']\downarrow$. 
\end{proof}

The following lemma expresses that the
term-relations  are closed on the right under arbitrary pure reduction
sequences and on the left under zero-step pure reduction sequences.
(The lemma could be strengthened slightly by allowing some of the
reductions to be non-pure, but the way it is stated
now, it holds both for the may interpretations of types, and also for
the must interpretations of types given in the following section.) 
\begin{lem}
 \label{lem:ext-lemma}
 For all $(e,e')\in \TT{\den{\tau}}_n$,
 \begin{itemize}
 \item if $e'_1\stepstopure e' \stepstopure e'_2$, then 
   $(e,e'_1)\in \TT{\den{\tau}}_n$ and
   $(e,e'_2)\in \TT{\den{\tau}}_n$;
 \item if $e_1\stepstozeropure e \stepstozeropure e_2$, then 
   $(e_1,e')\in \TT{\den{\tau}}_n$ and
   $(e_2,e')\in \TT{\den{\tau}}_n$.
 \end{itemize}
\end{lem}
The proof is straightforward; the use of $\stepstozeropure$ in the second
item ensures that the index of the relation does not change.

The relational interpretation extends pointwise to value substitutions: 
$(\gamma,\gamma')\in{\den{\Gamma}\vec r}_n$ if  $(\gamma(x),\gamma(x'))\in{\den{\tau}\vec r}_n$ for all $x{:}\tau\in\Gamma$. 
Based on this   interpretation we
consider the following type-indexed relation:
\begin{multline*}
\Delta;\Gamma\vdash e\maylogapprox \!e'\,{:}\,\tau\ \text{ where $\Delta = \vec\alpha$}\\ 
\IFF\ 
\forall\vec\tau,\vec\tau'\!.\, \forall\vec r\in\!\!\VRel_\omega(\vec\tau,\vec\tau').\, \forall n\mathop<\omega.\,\forall(\gamma,\gamma')\in{\den{\Gamma}\vec r}_n.\  (e\subst{\vec\alpha}{\vec\tau} \gamma, e'\subst{\vec\alpha}{\vec\tau'} \gamma')\in\TT{\den{\tau}\vec r}_n
\end{multline*}
The definition of $\maylogapprox$ builds in enough closure properties to prove its compatibility.

\begin{prop}
[Fundamental property]\ 
\label{prop:FTLR}
The relation $\maylogapprox$ has the compatibility properties given in Figure~\ref{fig:compatibility}. 
In particular, it is reflexive: if $\Delta;\Gamma\vdash e:\tau$ then $\Delta;\Gamma\vdash e\maylogapprox e:\tau$. 
\end{prop}
\begin{proof}
We consider the inference rules from Figure~\ref{fig:compatibility} in turn. 
\begin{itemize}

\item For the introduction of recursive types, we assume  that
\[
\Delta;\Gamma \vdash v\maylogapprox v':\tau_j\subst\alpha{\REC\alpha{\tau_1\mathop+\ldots\mathop+\tau_m}},\]
and then prove that 
\[
\Delta;\Gamma\vdash \IN{j}\,{v}\maylogapprox
\IN{j}\,{v'}:\REC\alpha{\tau_1\mathop+\ldots\mathop+\tau_m}.
\] 

For notational convenience we only consider the case of closed terms.
Let $\tau$ abbreviate the type
$\REC\alpha{\tau_1\mathop+\ldots\mathop+\tau_m}$.  Note that
\[ 
  \den{\tau}\vec r = \bigcup_j \IN{j}\bigl(\LATER\den{\tau_j}(\vec
r,\den{\tau}\vec r)\bigr) = \bigcup_j
\IN{j}\bigl(\LATER\den{\tau_j\subst\alpha\tau}(\vec r)\bigr)
\]
by
definition and Lemma~\ref{lem:may:subst}, and that the inclusion
$\den{\tau_j\subst\alpha\tau}(\vec
r)\subseteq\LATER\den{\tau_j\subst\alpha\tau}(\vec r)$ holds. It is easy to see, straight
from the definition, that $\left(\FUN{x}{\IN{j}\,x}, \FUN{x}{\IN{j}\,x}\right) \in
\den{\tau_j\subst{\alpha}{\tau} \to \tau} \vec{r}_n$, so
assuming $(E,E')\in\T{\den{\tau}\vec r}_n$ it follows from
Lemma~\ref{lem:context-composition} that
\[
(E[(\FUN{x}{\IN{j}\,x})\,[]],E'[(\FUN{x}{\IN{j}\,x})\,[]])\in\T{\den{\tau_j\subst\alpha\tau}\vec
 r}_{n}.
\]  
Thus, if $E[\IN{j}\,v]\downarrow_i$ for some $i\leq n$
then $E'[(\FUN{x}{\IN{j}\,x})\,v'])\downarrow$ follows from
$(v,v')\in\TT{\den{\tau_j\subst\alpha\tau}\vec r}_{n}$.  Therefore
we can conclude $E'[\IN{j}\,v']\downarrow$, and have shown
$(\IN{j}\,v,\IN{j}\,v')\in\TT{\den{\tau}\vec r}_n$.  Since $n$ was
chosen arbitrarily, we have $\Delta;\Gamma\vdash
\IN{j}\,{v}\maylogapprox \IN{j}\,{v'}:\tau$.

\item For the elimination of recursive types, we assume  that $\tau$ is of the form $\REC\alpha{\tau_1\mathop+\ldots\mathop+\tau_m}$, 
$\Delta;\Gamma,x_j{:}\tau_j\subst{\alpha}{\tau}\vdash e_j\maylogapprox e'_j:\tau'$  for all $1\leq j\leq m$  
and $\Delta;\Gamma\vdash v\maylogapprox v':\tau$. 
We prove $\Delta;\Gamma\vdash \CASE{v}(\ldots|\BRANCH j{x_j}{e_j}|\ldots)\maylogapprox \CASE{v'}(\ldots|\BRANCH j{x_j}{e_j'}|\ldots) :\tau'$. 

For simplicity we only consider the case of closed terms. 
By definition and by\break Lemma~\ref{lem:may:subst} we have
 $\den{\tau}\vec r = \bigcup_j \IN{j}\bigl(\LATER\den{\tau_j}(\vec r,\den{\tau}\vec r)\bigr) = \bigcup_j \IN{j} (\LATER\den{\tau_j\subst\alpha\tau}\vec r)$. 
 Moreover,\break 
$(\FUN x{\CASE x (\ldots\SEP\BRANCH j{x_j}{e_j}\SEP\ldots)}, \FUN x{\CASE x (\ldots\SEP\BRANCH j{x_j}{e_j'}\SEP\ldots)})\in\den{\tau\to\tau'}\vec r_n$ for any $n$. 
To see this, assume $k\leq n$, let $(a,a')\in\den{\tau}\vec r_n$ and $(E,E')\in\T{\den{\tau'}\vec r}_n$ 
such that $E[\CASE a (\ldots\SEP\BRANCH j{x_j}{e_j}\SEP\ldots)]\downarrow_k$. This implies
that $k > 0$ and by the above observation we have $a=\IN{j}{a_j}$ and $a'=\IN{j}{a_j'}$ for some $(a_j,a_j')\in\den{\tau_j\subst\alpha{\tau}}\vec r_{k-1}$. 
From $E[\CASE a (\ldots\SEP\BRANCH j{x_j}{e_j}\SEP\ldots)]\downarrow_k$ we obtain $E[e_j\subst{x_j}{a_j}]\downarrow_{k-1}$, and thus the assumption on $e_j$ and $e_j'$ gives $E'[e_j'\subst{x_j}{a_j'}]\downarrow$. This shows that $E'[\CASE {a'} (\ldots\SEP\BRANCH j{x_j}{e_j'}\SEP\ldots)]\downarrow$ holds. 

To prove the case, assume next  that $(E,E')\in\T{\den{\tau'}\vec r}_n$. 
From   Lemma~\ref{lem:context-composition} we obtain 
$(E[(\FUN x{\CASE x (\ldots\SEP\BRANCH j{x_j}{e_j}\SEP\ldots)})\,[]],E'[(\FUN x{\CASE x (\ldots\SEP\BRANCH j{x_j}{e_j'}\SEP\ldots)})\,[]])\in\T{\den{\tau}\vec r}_{n}$. 
Since by assumption $(v,v')\in\TT{\den{\tau}\vec r}_{n}$, we get that 
$E[\CASE{v}(\ldots|\BRANCH j{x_j}{e_j}|\ldots)]\downarrow_n$  
implies $E[\CASE{v'}(\ldots|\BRANCH j{x_j}{e_j'}|\ldots)]\downarrow$ as required. 

\item For choice, we assume $\Delta\vdash\Gamma$ and show $\Delta;\Gamma\vdash \CHOOSE\maylogapprox \CHOOSE:\NAT$. 
Suppose $(E,E')\in\T{\den{\NAT}\vec r}_n$ and $E[\CHOOSE]\downarrow_j$ for some $j\leq n$. 
Then $E[\CHOOSE]\stepsto E[\underline k]$ and $E[\underline k]\downarrow_{j}$ for some $k\in\N$. 
By induction on $k$ we obtain that $(\underline k,\underline k)\in\den{\NAT}\vec r_n$, 
and thus $E'[\underline k]\downarrow$. 
Hence $E'[\CHOOSE]\downarrow$. 
\end{itemize}
The proofs for  the remaining rules are similar. 
\end{proof}

\begin{cor}
  \label{cor:values-FTLR}
  If $v \in Val(\tau)$ then for all $n < \omega$, $(v,v) \in \den{\tau}_n$.
\end{cor}
\begin{proof}
  We prove this by induction on the value $v$. 

  \begin{itemize}
  \item Suppose $\tau = \tau_1 \to \tau_2$ and $v = \FUN{x}{e}$. Fix $n$ and let $i \leq
    n$. For arbitrary $(u, u') \in \den{\tau_1}_i$ we have to prove $\left(e\subst{x}{u},
      e\subst{x}{u'}\right) \in \TT{\den{\tau_2}_i}$. Since $\emp ; x\, {:}\, \tau_1
    \vdash e\, {:}\, \tau_2$ using Proposition~\ref{prop:FTLR} we have
    $\emp ; x\, {:}\, \tau_1 \vdash e \maylogapprox e \,{:}\,\tau_2$. If we instantiate
    this with $i$ and the substitution $x \mapsto (u, u')$ we get what is required.
  \item Suppose $\tau = \ALL{\alpha}{\tau}$ and $v = \TFUN{\alpha}{e}$. Fix $n$, pick
    $\tau, \tau' \in \Type$ and $R \in VRel_n(\tau, \tau')$. We have to show
    $\left(e\subst{\alpha}{\tau}, e\subst{\alpha}{\tau'}\right) \in \TT{\den{\tau}R_n}$,
    but this again follows straightforwardly from Proposition~\ref{prop:FTLR}.
  \end{itemize}

\noindent The other cases follows straightforwardly from the induction hypothesis. The case for
  $\IN{j}$ also requires Lemma~\ref{lem:may:subst}.
\end{proof}

\begin{thm}
[Coincidence] 
\label{thm:coincidence}
$\Delta;\Gamma\vdash e\maylogapprox e':\tau$  
if and only if 
$\Delta;\Gamma\vdash e\mayciu e':\tau$.   
\end{thm}
\begin{proof}
 For the direction from left to right, let $\delta\in\Type^\Delta$,
 $\gamma\in\Subst(\Gamma\delta)$ and $E\in\Ctxt(\tau\delta)$, and
 assume $E[e \delta\gamma]\downarrow$. Then  $E[e
 \delta\gamma]\downarrow_n$ for some $n$.  We must show $E[e'\delta
 \gamma]\downarrow$.  As a consequence of
 Proposition~\ref{prop:FTLR} and Corollary~\ref{cor:values-FTLR},
 $(\gamma,\gamma)\in\den{\Gamma\delta}_n$ and
 $(E,E)\in\T{\den{\tau\delta}}_n$.  By definition of
 $\Delta;\Gamma\vdash e\maylogapprox e':\tau$ and
 Lemma~\ref{lem:may:subst} we have $(e \delta\gamma,e'
 \delta\gamma)\in\TT{\den{\tau\delta}}_n$, and thus $E[e
 \delta\gamma]\downarrow_n$ gives $E[e' \delta\gamma]\downarrow$.

For the direction from right to left, first note that the logical relation is closed under may-CIU approximation; more precisely,  if 
$\Delta;\Gamma\vdash e\maylogapprox e':\tau$ and $\Delta;\Gamma\vdash e'\mayciu e'':\tau$ 
then $\Delta;\Gamma\vdash e\maylogapprox e'':\tau$. This observation follows from the definition of $\TT{(\cdot)}$ used in $\Delta;\Gamma\vdash e\maylogapprox e':\tau$  and the definition of CIU approximation.
Now assume that  $\Delta;\Gamma\vdash e\mayciu e':\tau$. 
By  Proposition~\ref{prop:FTLR}, $\Delta;\Gamma\vdash e\maylogapprox e:\tau$, 
and thus $\Delta;\Gamma\vdash e\maylogapprox e':\tau$. 
\end{proof}

\begin{figure}
\begin{align*}
\inferrule{\Delta;\Gamma\vdash v\mathrel\R v':\tau\\ \Delta;\Gamma,x{:}\tau\vdash e\mathrel\R e':\tau'}
{\Delta;\Gamma\vdash e\subst xv \mathrel\R e'\subst x{v'} :\tau'}
\qquad
\inferrule{\Delta,\alpha;\Gamma\vdash e\mathrel\R e':\tau'}
{\Delta;\Gamma\subst\alpha\tau\vdash e  \mathrel\R e'  :\tau'\subst\alpha\tau}\ \Delta\vdash\tau
\end{align*}
\caption{\label{fig:substitutivity} Substitutivity   properties of type-indexed relations}
\end{figure}

\begin{proof}[Proof of CIU Theorem~\ref{thm:ciu}(1)]
We first show that $\mayciu$ is contained in $\mayctxtapprox$. By definition, $\mayctxtapprox$ is the largest may-adequate precongruence, thus it is sufficient to establish that $\mayciu$ is a may-adequate precongruence. 
From the definition it is immediate that $\mayciu$ is may-adequate, reflexive and transitive. 
By Theorem~\ref{thm:coincidence}, $\mayciu$ coincides with $\maylogapprox$ which is compatible by Proposition~\ref{prop:FTLR}.

For the other direction, following Pitts \cite{Pitts:10}, we first consider the special case where $\emp;\emp\vdash e\mayctxtapprox e':\tau$. 
To prove $\emp;\emp\vdash e\mayciu e':\tau$, note that  $\emp;\emp\vdash E[e]\mayctxtapprox E[e']:\tau'$ holds for all evaluation contexts $E$ such that $\vdash E:\tau\sto\tau'$ since $\mayctxtapprox$ is reflexive and compatible. 
Hence, that  $E[e]\downarrow$ implies $E[e']\downarrow$ follows since $\mayctxtapprox$ is may-adequate. 

The general case reduces to this special case since may-contextual approximation has the substitutivity properties given in  Figure~\ref{fig:substitutivity}. 
For the first of these, assume $\Delta;\Gamma\vdash v \mayctxtapprox v':\tau$ and $\Delta;\Gamma,x{:}\tau\vdash e  \mayctxtapprox e':\tau'$. 
From the definition of may-CIU approximation it is easy to see  
\[
\Delta;\Gamma \vdash e\subst{x}{v}  \mayciu (\FUN{x}{e})\,v:\tau'
\quad\text{and}\quad 
\Delta;\Gamma \vdash (\FUN{x}{e' })\,v'  \mayciu {e' }\subst{x}{v'}:\tau'\ .
\]
Since we have already shown that $\mayciu$ is contained in $\mayctxtapprox$, and since 
$\Delta;\Gamma\vdash (\FUN{x}{e})\,v\mayctxtapprox (\FUN{x}{e'})\,v' :\tau'$ by compatibility, 
we can conclude $\Delta;\Gamma \vdash e\subst{x}{v}  \mayctxtapprox {e'}\subst{x}{v'}:\tau'$ by transitivity.
The second substitutivity property is proved similarly, using a weakening property of may-contextual approximation. 
\end{proof}

Using the logical relation, we now prove some simple extensionality
properties for may contextual approximation and equivalence.  We will
use these properties in the parametricity example in
Section~\ref{sect:applications}.

\begin{lem}
\label{lem:may:forall-apply}
If $v \in \Val\left(\ALL{\alpha}{\sigma}\right)$ then $\forall
\tau, \tau' \in \Type, \forall R \in \VRel\left(\tau, \tau'\right), \forall n < \omega,
(\TAPP{v}{\tau}, \TAPP{v}{\tau'}) \in \TT{\den{\sigma}R_n}$.
\end{lem}

\begin{proof}
Take $n < \omega, j \leq n, (E, E') \in \T{\den{\sigma}R}_j$ and
assume $E\left[\TAPP{v}{\tau}\right]\downarrow_j$ which is equivalent to 
$E[(\FUN{x}{\TAPP{x}{\tau}}) v]\downarrow_j$. It is easy to see that 
$\left(E\left[(\FUN{x}{\TAPP{x}{\tau}}) []\right], 
  E'\left[(\FUN{x}{\TAPP{x}{\tau'}}) []\right]\right)
\in\T{\den{\ALL{\alpha}{\sigma}}R}_j$ and using
Proposition~\ref{prop:FTLR} we have $\forall n < \omega, (v,v) \in
\TT{\den{\ALL{\alpha}{\alpha\times\alpha\to\alpha}}}_n$ which concludes the proof.
\end{proof}

\begin{lem}
[Application]
\label{lem:may:application}
If $(e,e')\in\TT{\den{\tau_1}\vec r}_n$ 
and $(v,v')\in\den{\tau_1\mathop\to\tau_2}\vec r_n$ 
then $(v\,e,v'\,e')\in\TT{\den{\tau_2}\vec r}_n$. 
\end{lem}
\begin{proof}
For any $(E,E')\in\T{\den{\tau_2}\vec r}_n$, 
$(E[v\,[]],E'[v'\,[]])\in\T{\den{\tau_1}\vec r}_n$ by Lemma~\ref{lem:context-composition}. 
Thus, if $E[v\,e]\downarrow_j$ for $j\leq n$ then $E'[v'\,e']\downarrow$.  
\end{proof}

\begin{lem}
\label{lem:may:forall-ext}
If $v, u \in \Val\left(\ALL{\alpha}{\sigma}\right)$, $n < \omega$ and 
$\forall \tau, \tau'\in\Type, \forall R \in \VRel(\tau, \tau'),
\left(\TAPP{v}{\tau}, \TAPP{u}{\tau'}\right) \in \TT{\den{\sigma}}R_n$
then $(v, u) \in \den{\ALL{\alpha}{\sigma}}_n$
\end{lem}

\begin{lem}
\label{lem:may:apply-let}
If $\tau, \sigma \in \Type$, $n < \omega$, $(f,f') \in \TT{\den{\tau\to\sigma}R_n}$ and
$(e,e') \in \TT{\den{\tau}R_n}$ then 
\[\left((\FUN{x}{x\,e})\,f, (\FUN{x}{x\,e'})\,f'\right) \in \TT{\den{\sigma}R_n}.\]
\end{lem}
\begin{proof}
This follows from Lemma~\ref{lem:context-composition} and Lemma~\ref{lem:may:application}.
\end{proof}

\begin{lem}[Functional extensionality]
\label{lem:may:fun-extensionality}
Let $\tau, \sigma \in \Type, f, g \in \Val(\tau\to\sigma)$ and assume
$\forall u \in \Val(\tau), f\,u \mayctxtequiv g\,u$. Then $f \mayctxtequiv g$. 
\end{lem}

\begin{proof}
We will show directly that $\forall n < \omega, (f, g) \in \den{\tau\to\sigma}_n$. To
this end take $n < \omega, j \leq n$ and $(v, u) \in \den{\tau}_j$. By the canonical forms
lemma $f = \FUN{x}{e}$ and $g=\FUN{x}{e'}$ for some $e$ and $e'$. We must show
$(e\subst{x}{v}, e'\subst{x}{u})\in\TT{\den{\sigma}}_j$. So take a $k\leq j$, $(E,E')
\in \T{\den{\sigma}}_k$ and assume $E[e\subst{x}{v}]\downarrow_k$ which is equivalent to
$E[f\,v]\downarrow_k$. Proposition~\ref{prop:FTLR} shows that $\forall m < \omega, (f,f)
\in \TT{\den{\sigma\to\tau}}_m$ and then using Lemma~\ref{lem:may:apply-let} we can
conclude that $E'[f\,u]\downarrow$. Using the assumption and
Theorem~\ref{thm:coincidence} we get $E'[g\,u]\downarrow$, which concludes the proof.
\end{proof}

Note that extensionality property above is stated for \emph{values} of
function type; a more general extensionality property for \emph{expressions} of function
type fails.  To show that, we first define some abbreviations. 

\newcommand{\letapp}[2]{\ensuremath{(\FUN{x}{x\,#2})\,#1}}

\newcommand{\B}{\mathbf{2}}
\newcommand{\tval}{\mathbf{true}}
\newcommand{\fval}{\mathbf{false}}
\newcommand{\IFT}[3]{\syn{if}\,#1\,\syn{then}\,#2\,\syn{else}\,#3}

Let $\B$ be the type $\REC{\alpha}{\ONE + \ONE}$
and let $\tval = \IN{1}\unitval$ and $\fval = \IN{2}\unitval$ be values of type $\B$. By the
canonical forms lemma these two are the only closed values of this type. Let $\Omega_{\B}
= \TAPP{\Omega}{\B}$. Note that reduction from $\Omega_{\B}$ is deterministic and
non-terminating. We first define $\syn{if}$ and $\syn{ifz}$ constructs
as
\begin{align*}
\IFT{p}{e}{e'} &= \letr{y}{p}{\CASE{p}\;\BRANCH{1}{x}{e}\SEP\BRANCH{2}{x}{e'}}\\
\IFZERO{p}{e}{e'} &= \letr{y}{p}{\CASE{\underline n}\;\BRANCH{1}{x}{e}\SEP\BRANCH{2}{x}{e'}}
\end{align*}
where $x$ and $y$ are variables not free in $e$ or $e'$.

Now, to exhibit the failure of a more general extensionality property,
let $e\mkern-4 mu = \mkern -4 mu\FUN{x}{\syn{proj}_1\,x\mkern-2 mu \OR\mkern -2 mu \syn{proj}_2\,x}$ and $e' =
\IFZERO{\CHOOSE}{\FUN{x}{\syn{proj}_1\,x}}{\FUN{x}{\syn{proj}_2\,x}}$ be two terms of
type $\B\times\B\to\B$.

Then it is easy to see that $\forall u \in \Val(\B\times\B), \letapp{e}{u} \mayctxtequiv
\letapp{e'}{u}$. But on the other hand there is an evaluation context distinguishing the
two terms $e$ and $e'$. The idea is to call the resulting value twice with the same pair
and diverge if it produces the same value twice, but return a value if the results of
the two calls differ. To this end we first define the function $\syn{xor} =
\FUN{x}{\FUN{y}{\IFT{x}{(\IFT{y}{\fval}{\tval})}{y}}}$ and then we define
\begin{align*}
  E =\, &\letr{x}{[]}{\\
    &\letr{y}{x\,\PAIR{\tval}{\fval}}{\\
      &\letr{z}{x\,\PAIR{\tval}{\fval}}{\\
        &\letr{w}{x\, \syn{xor}\, y}{\\&\IFT{w}{w}{\Omega_{\B}}}}}}
\end{align*}
We then have $E[e]\downarrow$ but on the other hand, $E[e']$ always diverges, therefore
$e$ and $e'$ are not contextually equal.

A similar counter-example can also be exhibited for must-contextual equivalence.
Indeed, if we define the function 
$\syn{xnor} =
\FUN{x}{\FUN{y}{\IFT{x}{y}{(\IFT{y}{\fval}{\tval})}}}$ and the evaluation context
\begin{align*}
  E' =\, &\letr{x}{[]}{\\
    &\letr{y}{x\,\PAIR{\tval}{\fval}}{\\
      &\letr{z}{x\,\PAIR{\tval}{\fval}}{\\
        &\letr{w}{x\, \syn{xnor}\, y}{\\&\IFT{w}{w}{\Omega_{\B}}}}}}.
\end{align*}
we have that $\forall u \in \Val(\B\times\B), \letapp{e}{u}\mustctxtequiv
\letapp{e'}{u}$, but on the other hand $E'[e']\Downarrow$ but not $E'[e]\Downarrow$.

Finally, we state the expected extensionality property for values of
polymorphic type. 

\begin{lem}[Extensionality for $\forall$]
\label{lem:may:extensionality-ctx-forall}
Let $u, v \in \Val\left(\ALL{\alpha}{\sigma}\right)$ and assume $\forall \tau\in\Type,
\TAPP{u}{\tau} \mayctxtequiv \TAPP{v}{\tau}$. Then $u \mayctxtequiv v$.
\end{lem}

The proof of this lemma is essentially the same as the proof of
Lemma~\ref{lem:may:fun-extensionality}.

\section{Must equational theory}
\label{sect:musttheory}
To define the logical relation for must-approximation, we need to stratify the observables again. 
We define stratified relations counting all steps ($\lightning$) and one counting
only unfold-fold reductions ($\Downarrow$).  The latter is used for
indexing the logical relations, the former for relating the latter to must-approximation.

For   terms $e$ and ordinals $\beta$ we define 
$e\lightning_\beta$ by induction on $\beta$:
$e\lightning_\beta$ if 
for all $e'$ such that $e\stepsto e'$ there exists $\nu<\beta$ and $e'\lightning_\nu$. 
The essential observation is that  $\lightning_\beta$ indeed captures must-convergent behaviour. 

\begin{lem}
[Stratified must-convergence]
\label{prop:stratified-must-convergence}
$e\Downarrow$ if and only if $e\lightning_\beta$ for some $\beta<\omega_1$ (for $\omega_1$ the least uncountable ordinal).
\end{lem}
\begin{proof}
The proof from left to right is by induction on $e\Downarrow$. 
By induction hypothesis there exists ordinals $\nu(e')<\omega_1$ for each term $e'$ such that $e\stepsto e'$.
Let $\beta=\bigcup \nu(e')$, then $\beta\mathop+1<\omega_1$ (since there are only countably many such $e'$ and each $\nu(e')$ is countable) and $e\lightning_{\beta\mathop+1}$. 
The direction from right to left is by induction on $\beta$. 
\end{proof}

For   terms $e$ and ordinals $\beta$ we define 
$e\Downarrow_\beta$ by induction on $\beta$:
$e\Downarrow_\beta$ if 
for all $e'$ such that $e\stepstoone e'$ there exists $\nu<\beta$ and $e'\Downarrow_\nu$. 

Using Lemma~\ref{prop:stratified-must-convergence}, we can show:
\begin{lem}
 \label{lem:11b}
 $e\Downarrow$ implies $e\Downarrow_\beta$ for some $\beta<\omega_1$.
\end{lem}

\paragraph*{Logical $\omega_1$-indexed uniform relation for must-approximation}

Proposition~\ref{prop:stratified-must-convergence} indicates that   logical relations for  must-approximation need to be indexed over  $\omegaOne$. The lifting of value relations $r\in\VRel_ \omegaOne(\tau,\tau')$ 
to relations $\T r\in\SRel_ \omegaOne(\tau,\tau')$ on evaluation contexts and to relations
$\TT r\in\ERel_ \omegaOne(\tau,\tau')$ on terms is defined with respect to must termination.
\begin{align*}
\T{r}_\beta &= \{(E,E')\;|\;\forall \nu\leq \beta.\ \forall (v,v')\in r_\nu.\ E[v]\mathrel{\Downarrow_\nu}\ \Rightarrow\ E'[v']\mathrel\Downarrow\ \}\\
\TT{r}_\beta &= \{(e,e')\;|\;\forall \nu\leq \beta.\ \forall (E,E')\in\T{r}_\nu.\ E[e]\Downarrow_\nu\ \Rightarrow\ E'[e']\Downarrow\ \}
\end{align*}
Except for this difference, the relational interpretation $\den{\tau}({\vec
 r})\in\VRel_\omegaOne(\tau\subst{\vec\alpha}{\vec\tau},\tau\subst{\vec\alpha}{\vec\tau'})$
is literally the same as in Section~\ref{sect:maytheory} and defined by induction on the type $\vec\alpha\vdash\tau$, given 
closed types $\tau_1,\tau_1',\ldots,\tau_k,\tau_k'\in\Type$ and
relations
$r_1\in\VRel_ \omegaOne(\tau_1,\tau_1'),\ldots,r_k\in\VRel_ \omegaOne(\tau_k,\tau_k')$: 
\begin{align*}
\LR{\alpha_i}(\vec r) &= r_i &
\LR{\tau_1\times\tau_2}(\vec r) &= \LR{\tau_1}(\vec r)\mathop\times\den{\tau_2}(\vec r)\\
\LR{\ONE}(\vec r) &= (\textit{Id}_\ONE)_{\beta<\omegaOne}&
\LR{\tau_1\to\tau_2}(\vec r) &=  \LR{\tau_1}(\vec r)\mathop\to\TT{\LR{\tau_2}(\vec r)}\\
\LR{\ALL\alpha{\tau}}(\vec r) &= \forall r.\TT{\den{\tau}(\vec r,r)} &
\den{\REC\alpha{\tau_1\mathop+\ldots\mathop+\tau_m}}(\vec r)&= \textit{fix}\, s.\,\textstyle{\bigcup_j} \IN{j}(\LATER \den{\tau_j}(\vec r,s))
\end{align*}
Logical must-approximation is defined as follows:
\begin{multline*}
\Delta;\Gamma\vdash e\mustlogapprox \!e'\,{:}\,\tau\ \text{ where $\Delta=\vec\alpha$}\\ 
\IFF\ 
\forall\vec\tau,\vec\tau'\!.\,\forall\vec r\in\!\!\VRel_\omegaOne(\vec\tau,\vec\tau').\, \forall \beta\mathop<\omegaOne.\,\forall(\gamma,\gamma')\in{\den{\Gamma}\vec r}_\beta.\ (e \subst{\vec\alpha}{\vec\tau}\gamma, e'\subst{\vec\alpha}{\vec\tau'} \gamma')\in\TT{\den{\tau}\vec r}_\beta
\end{multline*}

\begin{prop}
[Fundamental property]\ 
\label{prop:must-FTLR}
The relation $\mustlogapprox$ has the compatibility properties given in Figure~\ref{fig:compatibility}. 
In particular, it is reflexive: if $\Delta;\Gamma\vdash e:\tau$ then $\Delta;\Gamma\vdash e\mustlogapprox e:\tau$. 
\end{prop}

\begin{proof}
 The proof is similar to the one for Proposition~\ref{prop:FTLR}. We
 give only the case for choice, where we assume $\Delta\vdash \Gamma$
 and prove $\Delta;\Gamma\vdash \CHOOSE\mustlogapprox \CHOOSE:\NAT$.
 Suppose $(E,E')\in\T{\den{\NAT}\vec r}_\beta$ and
 $E[\CHOOSE]\Downarrow_\beta$.  
 We are to show that $E'[\CHOOSE]\Downarrow$, for which it suffices
 to show that $E'[\underline k]\Downarrow$, for all $k\in\N$.
 Let $k\in\N$ be arbitrary. Then
 $E[\CHOOSE]\stepsto E[\underline k]$ and so $E[\underline k]\Downarrow_\beta$.
 Induction on $k$ shows that $(\underline k,\underline
 k)\in\den{\NAT}\vec r_{\beta}$ and hence the required follows
 by the assumption on $(E,E')$.
\end{proof}

\begin{cor}
  \label{cor:must-values-FTLR}
  If $v \in Val(\tau)$ then for all $\nu < \omega_1$, $(v,v) \in \den{\tau}_{\nu}$.
\end{cor}
We omit the proof, as it is analogous to the proof of Corollary~\ref{cor:values-FTLR}.

\begin{thm}
[Coincidence] 
\label{thm:must-coincidence}
$\Delta;\Gamma\vdash e\mustlogapprox e':\tau$  
if and only if 
$\Delta;\Gamma\vdash e\mustciu e':\tau$.   
\end{thm}
\begin{proof}
 The proof is completely analogous to that of
 Theorem~\ref{thm:coincidence}.  For the direction from left to right
 one uses the relationship between $\Downarrow$ and
 $\Downarrow_\beta$ given by Lemma~\ref{lem:11b} and then
 appeals to Proposition~\ref{prop:must-FTLR} and Corollary~\ref{cor:must-values-FTLR}.
 The direction from right to left uses the fact that $\mustlogapprox$ is closed under
 must-CIU approximation.
\end{proof}
\begin{proof}[Proof of CIU Theorem~\ref{thm:ciu}(2)] 
The proof is analogous to that of  Theorem~\ref{thm:ciu}(1). 
From the definition, $\mustciu$ is a must-adequate reflexive and transitive relation, 
by Proposition~\ref{prop:must-FTLR} and Theorem~\ref{thm:must-coincidence} it is also compatible, and thus contained in $\mustctxtapprox$. 
From this containment and the closure of $\mustciu$ under  beta conversion it follows that $\mustctxtapprox$ has the substitutivity properties in Figure~\ref{fig:substitutivity}.  
Thus it suffices to prove the containment of $\mustctxtapprox$ in $\mustciu$ for closed terms, which is clear by the compatibility and must-adequacy of $\mustctxtapprox$. 
\end{proof}

\section{Applications}
\label{sect:applications}

This section illustrates how the logical relation characterization of
contextual approximation can be used to derive interesting examples
and further proof principles.  We consider three such applications: a
recursion-induction principle for recursively defined functions,
syntactic minimal invariance of a recursive type, and an application
of relational parametricity to characterize the elements of the type 
$\ALL{\alpha}{\alpha\times\alpha\to\alpha}$.

\paragraph*{Recursion-induction}

\begin{figure}
\begin{align*}
\inferrule{\Delta;\Gamma\vdash v\,v' \mayctxtapprox v':\tau_1\mathop\to\tau_2}{\Delta;\Gamma\vdash \TAPP{\FIX}{\tau_1\tau_2}\,v\mayctxtapprox v':\tau_1\mathop\to\tau_2}
\qquad
\inferrule{\Delta;\Gamma\vdash v\,v' \mustctxtapprox v':\tau_1\mathop\to\tau_2}{\Delta;\Gamma\vdash \TAPP{\FIX}{\tau_1\tau_2}\,v\mustctxtapprox v':\tau_1\mathop\to\tau_2}
\end{align*}
\caption{\label{fig:recursion-induction} Recursion induction: least prefixed point property of \FIX}
\end{figure}

Recall from the introduction that $\FIX : \ALL{\alpha,\beta}{((\alpha\mathop\to\beta)\mathop\to(\alpha\mathop\to\beta)) \to(\alpha\mathop\to\beta)}$ is given by the term $\TFUN{\alpha,\beta}{\FUN{f}{\delta_f(\IN{}\,\delta_f)}}$ where $\delta_f$ is an abbreviation for the term $\FUN{y}{\CASE{y}~\BRANCH{}{y'}{f(\FUN{x}{(\FUN r{r\,x})(y'\,y)})}}$. 
We now prove that $\FIX$ is a \emph{least} prefixed point
combinator. More precisely, we prove
(1) the soundness of the recursion-induction rules in
Figure~\ref{fig:recursion-induction};
and (2) Proposition~\ref{prop:fix}, which says that
$\FIX$ behaves as a fixed point combinator for
a large class of functionals, including all those of the form
$\FUN{f}{u}$, for $u$ a value. 
(Observe that this class of functionals includes those needed
for defining a standard fixed point expression $\FIX \,f(x).\, e$ via
an application of the fixed point combinator.)

Our recursion-induction rules are mild generalizations of the rules
given by Lassen~\cite{Lassen:PhD} 
who proved similar results (for a language without
polymorphism), when $v$ was restricted to be of the form
$\FUN{f}{u}$, for some value $u$.

We only include the proof for $\mustctxtapprox$ and for notational simplicity we
assume that the contexts $\Delta$ and $\Gamma$ are empty.
We assume the premise of the rule, and to show the conclusion we first prove that 
$(h,v')\in{\den{\tau_1\to\tau_2}}_\beta$ where $h$ is $\FUN{x}{(\FUN r{r\,x})\,(\delta_v\,(\IN{}\,\delta_v))}$, for all $\beta<\omegaOne$. 
The result then follows from the agreement of the logical relation with contextual approximation and transitivity, 
since $\TAPP{\FIX}{\tau_1}{\tau_2}\,v \ctxtequiv v\,h\mustctxtapprox v\,v'\mustctxtapprox v'$. 

To prove 
$(h,v')\in{\den{\tau_1\to\tau_2}}_\beta$ 
we proceed by induction on $\beta$
and assume that $(h, v')\in {\den{\tau_1\to\tau_2}}_{\nu}$, for all $\nu<\beta$;
we are then to show that $(h, v')\in {\den{\tau_1\to\tau_2}}_{\beta}$. 
From the canonical forms lemma it follows that $v'$ must be of the form $\FUN{x}{e'}$ for some $e'$. 
So let $\beta_1\leq\beta$ and $(u,u')\in\den{\tau_1}_{\beta_1}$, then it remains to show 
$((\FUN{r}{r\,u})(\delta_v\,(\IN{}\,\delta_v)), e'\subst{x}{u'})\in\TT{\den{\tau_2}}_{\beta_1}$. 

Suppose $\beta_2\leq\beta_1$, $(E,E')\in\T{\den{\tau_2}}_{\beta_2}$ and
$E[(\FUN{r}{r\,u})(\delta_v\,(\IN{}\,\delta_v))]\Downarrow_{\beta_2}$; we are to show $E'[e'\subst{x}{u'}]\Downarrow$. 
By (the must-analogue of) Lemma~\ref{lem:context-composition} and the fundamental property of the logical relation applied to $v$ we obtain  
$(E[(\FUN{r}{r\,u})\,((\FUN{x}{v\,x})\,[])], E'[(\FUN r{r\,u'})\,((\FUN{x}{v\,x})\,[])])\in\T{\den{\tau_1\to\tau_2}}_{\beta_2}$. 
Then, since $\delta_v\,(\IN{}\,\delta_v)\stepstoone v\,h$ and $(\FUN{x}{v\,x})\,h\stepsto v\,h$, we have 
$E[(\FUN{r}{r\,u})(v\,h)]\Downarrow_{\beta_3}$ for $\beta_3<\beta_2\leq\beta$, 
and hence also 
$E' [(\FUN r{r\,u'})\,(v\,v')]\Downarrow$ by induction hypothesis. 

By the premise and Theorem~\ref{thm:must-coincidence} we have that
$v\, v'$ CIU-approximates $v'$, and thus we get 
$E' [(\FUN r{r\,u'})\,v']\Downarrow$. 
Finally, since $(\FUN r{r\,u'})\,v'\stepsto^* e'\subst{x}{u'}$ we obtain the required 
$E'[e'\subst{x}{u'}]\Downarrow$.

\begin{prop}
  \label{prop:fix}
  Let $\tau, \tau' \in \Type$, $f \in \Val((\tau \to \tau') \to \tau \to \tau')$.
  If for all $g \in \Val(\tau\to\tau')$ there exists a value $f_g$ such that $f\,g
  \stepstopure f_g$, then $\TAPP{\TAPP{\FIX}{\tau}}{\tau'}\,f$ is a
  fixed point of $f$, i.e., 
  $f(\TAPP{\TAPP{\FIX}{\tau}}{\tau'}\,f) \mustctxtequiv
  \TAPP{\TAPP{\FIX}{\tau}}{\tau'}\,f$
and 
  $f(\TAPP{\TAPP{\FIX}{\tau}}{\tau'}\,f) \mayctxtequiv
  \TAPP{\TAPP{\FIX}{\tau}}{\tau'}\,f$.\footnote{We have abused notation slightly by
    writing $\TAPP{\TAPP{\FIX}{\tau}}{\tau'}\,f$ instead of
    $\letr{x}{\TAPP{\FIX}{\tau}}{\letr{y}{\TAPP{x}{\tau'}}{y\,f}}$, but the former is more
    readable.}
\end{prop}
\begin{proof}
  Fix $\tau, \tau' \in \Type$ and $f \in \Val((\tau \to \tau') \to \tau \to \tau')$.
  Let $\FIX_{\tau, \tau'} = \FUN{f}{\delta_f \, \IN{}{\delta_f}}$. It is then easy to see
  that $\FIX_{\tau, \tau'} f \stepstopure f\,h$ where 
  $h = \FUN{x}{(\FUN{r}{r\,x})(\FIX_{\tau,\tau'}\, f)}$. Similarly, $f\,
  (\FIX_{\tau,\tau'}\,f) \stepstopure f\,(f\,h)$.

  Let $E \in \Stk(\tau')$ and $v \in \Val(\tau)$. We have the following sequence of
  equivalences
  \begin{align*}
    E[h\,v]\Downarrow &\iff E[(\FUN{r}{r\,v})\,(\FIX_{\tau,\tau'}\,f)]\Downarrow\\
                      &\iff E[(\FUN{r}{r\,v})\,(f\,h)]\Downarrow\\
                      &\iff \left(\forall f', f\,h \stepsto^* f' \IMPLIES E[f'\,v]\Downarrow\right)\\
  \end{align*}
  and similarly, for may equivalence,
  \begin{align*}
    E[h\,v]\downarrow &\iff E[(\FUN{r}{r\,v})\,(\FIX_{\tau,\tau'}\,f)]\downarrow\\
                      &\iff E[(\FUN{r}{r\,v})\,(f\,h)]\downarrow\\
                      &\iff \left(\exists f'_v, f\,h \stepsto^* f'_v \land E[f'\,v]\downarrow\right)\\
  \end{align*}

  Suppose further that $f\,h \stepstopure f_h$, for some value $f_h$. 
  This implies that $f\,h$ reduces to a
  unique value. Then the above equivalences reduce to
  $E[h\,v]\Downarrow \iff E[f_h\,v]\Downarrow$ and $E[h\,v]\downarrow \iff
  E[f_h\,v]\downarrow$. Since $v$ and $E$ were arbitrary, we can use
  Lemma~\ref{lem:may:fun-extensionality} and its must-analogue to conclude that $h$ is may
  and must equivalent to $f_h$ and thus to $f\,h$, which concludes the proof.
\end{proof}

\paragraph*{Syntactic minimal invariance}
Consider the type $\tau = \REC\alpha{\NAT+\alpha\to\alpha}$. 
Let $\id=\FUN xx$ and consider the term 
\[
f\ \equiv\ \FUN{h,x}{\CASE{x}~{\BRANCH{1}{y}{\IN{1}\,{y}~\SEP~\BRANCH{2}{g}{\IN{2}\,{\FUN y{h(g(h\,y))}}}}}}\ .\]
We shall show that 
$\TAPP{\FIX}{\tau\tau}\,f\ctxtequiv \id :\tau\to\tau$. 
This equivalence corresponds to the characterization of solutions to recursive domain equations as minimal invariants in  domain-theoretic  work \cite{Pitts:96}, from which 
Pitts derives several (co-)induction principles.
Our proof is similar to the one by Dreyer, Ahmed, and
Birkedal~\cite{BirkedalL:lslr-journal} for a
language without nondeterminism.

By the soundness of the call-by-value beta- and eta-laws for contextual equivalence (Figure~\ref{fig:basictheory}) and the transitivity of  $\ctxtapprox$, it is easy to see that 
$f\,\id\ctxtequiv \id :\tau\to\tau$. 
The recursion-induction principle therefore yields  
$\TAPP{\FIX}{\tau\tau}\,f\ctxtapprox \id :\tau\to\tau$. 

For the reverse approximation we first show $\id\mustlogapprox h:\tau\to\tau$  where $h$ is again the term $\FUN{x}{(\FUN{r}{r\,x})(\delta_f\,(\IN{}\,\delta_f))}$. We show this 
by proving $(\id,h)\in\den{\tau\to\tau}_\beta$ for all $\beta<\omegaOne$ by induction on $\beta$. 
(The case for may-approximation is similar.) 

Thus it suffices to show,
for all $\nu\leq\beta$, for all $(v,v')\in\den{\tau}_\nu$, $(\id\, v,
h\, v')\in\TT{\den{\tau}}_\nu$.

Since $\den{\tau} = \IN{1}(\LATER\den{\NAT})\cup\IN{2}(\LATER\den{\tau\to\tau})$ there are two cases to consider:
\begin{itemize}
\item Case $(v,v')\in\IN{1}(\LATER\den{\NAT})_\nu$. 
 Then there exist $u,u'\in\Val(\NAT)$ such that $v=\IN{1}\,u$,
 $v'=\IN{1}\,u'$ and $(u,u')\in\den{\NAT}_{\nu'}$, for all
 $\nu'<\nu\leq\beta$. 
 Given $(E,E')\in\T{\den{\tau}}_\nu$ such that 
 $E[\id\, v]\Downarrow_\nu$, it suffices to show that
 $E[h\, v']\Downarrow$, which follows using the must-analogues of Lemmas~\ref{lem:ext-lemma}
 and~\ref{lem:extensiveness}
 since $h\, v'\stepstopure v'$ and $(v,v')\in\den{\tau}_\nu$ by
 assumption.
\item Case $(v,v')\in\IN{2}(\LATER\den{\tau\to\tau})_\nu$.  Then there
 exist $g,g'\in\Val(\tau\mathop\to\tau)$ such that $v=\IN{2}\,g$,
 $v'=\IN{2}\,g'$ and $(g,g')\in\den{\tau\mathop\to\tau}_{\nu'}$ for
 all $\nu'<\nu\leq\beta$.
 We are to show that $(\id\, v, h\, v')\in\TT{\den{\tau}}_\nu$.
 Since $h\, v'\stepstopure \IN{2}(\FUN{y}{h(g'(h\, y))})$ and
 $\id\, v\stepstozeropure v$, by must-analogue of Lemma~\ref{lem:ext-lemma}, it suffices to
 show
 $(\IN{2}(g), \IN{2}(\FUN{y}{h(g'(h\, y))}))\in \TT{\den{\tau}}_\nu$.
 Hence it suffices to show
 $(g, \FUN{y}{h(g'(h\, y))})\in \TT{\den{\tau\to\tau}}_{\nu_1}$, for
 all $\nu_1<\nu$.
 Pick $\nu_1< \nu$. 
 By the must-analogue of Lemma~\ref{lem:ext-lemma} it suffices to show, for all
 $\nu_2\leq \nu_1$,
 \begin{displaymath}
   \forall (u,u')\in\den{\tau}_{\nu_2}.\;
   (g(u), h(g'(h(u'))))\in\TT{\den{\tau}}_{\nu_2}.
 \end{displaymath}
 To this end, let $(u,u')\in\den{\tau}_{\nu_2}$ and pick $\nu_3\leq \nu_2$
 and suppose that $(E,E')\in\T{\den{\tau}}_{\nu_3}$.
 By the induction hypothesis and the must-analogue of Lemma~\ref{lem:ext-lemma},
 we get $(u, h(u'))\in \TT{\den{\tau}}_{\nu_3}$.
 Hence by the must-analogue of Lemma~\ref{lem:may:apply-let},
 we get 
 $(g(u), g'(h(u')))\in\TT{\den{\tau}}_{\nu_3}$,
 and thus it suffices to show that
 $(E,E'[h\,[]])\in \T{\den{\tau}}_{\nu_3}$.
 So let $\nu_4\leq \nu_3$ and take $(w,w')\in \den{\tau}_{\nu_4}$.
 We are to show that 
 if $E[w]\Downarrow_{\nu_4}$, then $E'[h\, w']\Downarrow$.
 By the induction hypothesis and the must-analogue of Lemma~\ref{lem:ext-lemma},
 $(w, h\, w')\in \TT{\den{\tau}}_{\nu_4}$, from which the required
 follows by the assumption on $(E,E')$.
\end{itemize}
By Theorem~\ref{thm:must-coincidence} and the CIU theorem,
$\id\mustlogapprox h:\tau\to\tau$ implies $\id\mustctxtapprox
h:\tau\to\tau$.  Since $\id\ctxtequiv f\,\id:\tau\to\tau$ and
$f\,h\ctxtequiv \TAPP{\FIX}{\tau\tau}\,f:\tau\to\tau$ we obtain
$\id\mustctxtapprox\TAPP{\FIX}{\tau\tau}\,f:\tau\to\tau$ by
compatibility and transitivity of must-contextual equivalence.

\paragraph*{Parametricity}

We will now characterize the elements of the type
$\ALL{\alpha}{\alpha\times\alpha\to\alpha}$, using relational 
parametricity. The main result is expressed as Theorem~\ref{thm:parametricity};
we first start with some lemmas. 
We only state and prove the results for must-contextual equivalence;
for may-contextual equivalence the properties and proofs are analogous.

\begin{lem}
\label{lem:must:ex:div1}
Let $v \in \Val\left(\ALL{\alpha}{\alpha\times\alpha\to\alpha}\right)$. If there exists
a $\tau\in\Type$ such that $\TAPP{v}{\tau}$ may-diverges then $v \mustctxtequiv
\left(\TFUN{\alpha}{\TAPP{\Omega}{(\alpha \times \alpha \to \alpha)}}\right)$.
\end{lem}

\begin{proof}
Let $\tau$ be such that $\TAPP{v}{\tau}$ may-diverges. By the must-analogue of
Lemma~\ref{lem:may:forall-apply} we have $\forall \tau'\in\Type, \forall R,\forall \nu <
\omega_1$, $\left(\TAPP{v}{\tau'}, \TAPP{v}{\tau}\right) \in
\TT{\den{\alpha\times\alpha\to\alpha}R_{\nu}}$. This implies that for all $\tau' \in \Type$,
$\TAPP{v}{\tau'}$ may-diverges (because the empty context is always related to itself,
for instance).

Using the must-analogue of Lemma~\ref{lem:may:forall-ext} we can thus conclude that $v
\mustlogapprox \TFUN{\alpha}{\TAPP{\Omega}{(\alpha \times \alpha \to \alpha)}}$
and $\TFUN{\alpha}{\TAPP{\Omega}{(\alpha \times \alpha \to \alpha)}} \mustlogapprox
v$. Theorems \ref{thm:must-coincidence} and \ref{thm:ciu} then finish the proof.
\end{proof}

\begin{lem}
\label{lem:must:ex:div2}
Let $v \in
\Val\left(\ALL{\alpha}{\alpha\times\alpha\to\alpha}\right)$. If for all
$\tau \in \Type$, the expression $\TAPP{v}{\tau}$ must-converges, and
there exist a $\tau$ and $u
\in\Val(\tau\times\tau)$ such that $\letapp{(v\,\tau)}{u}$ may-diverges, then for all
$\tau' \in Type$ and for all $u' \in Val(\tau' \times \tau')$, $\letapp{(v\,\tau')}{u'}
 \mustctxtequiv \TAPP{\Omega}{\tau'}$.
\end{lem}

\begin{proof}
Let $\tau' \in \Type, u' \in \Val\left(\tau' \times \tau'\right)$. By the canonical forms
lemma $u = \PAIR{u_1}{u_2}$ for some $u_1, u_2 \in \Val(\tau)$ and $u' =
\PAIR{u_1'}{u_2'}$ for some $u_1', u_2' \in \Val(\tau')$. Let $R_{\nu} = \{(u_1', u_1), (u_2',
u_2) \}$ for $\nu < \omega_1$. It is easy to see that $(u', u) \in \den{\alpha \times \alpha} R_{\nu}$.
The must-analogues of Lemmas \ref{lem:may:apply-let} and \ref{lem:may:forall-apply} then imply that
$(\letapp{(v\,\tau')}{u'}, \letapp{(v\,\tau)}{u}) \in \TT{\den{\alpha}R_{\nu}}$.
This in particular means that $\letapp{(v\,\tau')}{u'}$ may-diverges.
Since $\tau'\in\Type$ and $u'\in\Val(\tau'\times\tau')$ were arbitrary, we have that for all $\tau'\in\Type$
and $u' \in \Val(\tau' \times\tau')$, $\letapp{(v\,\tau')}{u'} \mustctxtequiv \TAPP{\Omega}{\tau'}$.

\end{proof}

\newcommand{\uval}{t}
\newcommand{\vval}{s}

\begin{lem}
\label{lem:must:ex:converging}
Let $v \in \Val\left(\ALL{\alpha}{\alpha\times\alpha\to\alpha}\right)$. If for all 
$\tau \in \Type$ and for all $u\in\Val(\tau\times\tau)$, the
expression $\letapp{(\TAPP{v}{\tau})}{u}$
must-converges, one of the following three cases holds
\begin{enumerate}
\item $\forall \tau\in\Type,\forall \uval, \vval \in \Val(\tau), 
  \letapp{(\TAPP{v}{\tau})}{\PAIR{\uval}{\vval}} \mustctxtequiv \uval$
\item $\forall \tau\in\Type,\forall \uval, \vval \in \Val(\tau), 
  \letapp{(\TAPP{v}{\tau})}{\PAIR{\uval}{\vval}} \mustctxtequiv \vval$
\item $\forall \tau\in\Type,\forall \uval, \vval \in \Val(\tau), 
  \letapp{(\TAPP{v}{\tau})}{\PAIR{\uval}{\vval}} \mustctxtequiv \uval \OR \vval$
\end{enumerate}
\end{lem}

\begin{proof}
Let $\tau\in\Type$. Must-analogues of Lemmas~\ref{lem:may:forall-apply}, \ref{lem:may:apply-let} and the
definitions of relational actions show that 
\begin{align}
  \forall \nu < \omega_1, \forall R \in \VRel(\B, \tau), \forall (b,w) \in (R\times R)_{\nu},
  (\letapp{(\TAPP{v}{\B})}{b}, \letapp{(\TAPP{v}{\tau})}{w})\in \TT{R_{\nu}} \label{eq:lem:must:ex:converging-1}
  \intertext{and}
  \forall \nu < \omega_1, \forall S \in \VRel(\tau, \B), \forall (w,b) \in (S\times S)_{\nu},
  (\letapp{(\TAPP{v}{\tau})}{w}, \letapp{(\TAPP{v}{\B})}{b})\in \TT{S_{\nu}}. \label{eq:lem:must:ex:converging-2}
\end{align}
By assumption there exists a $\vval\in\Val(\B)$, such that
$\letapp{(\TAPP{v}{\tau})}{\PAIR{\tval}{\fval}} \stepsto^* \vval$. By the canonical forms
lemma, $\vval$ can only be $\tval$ or $\fval$ and based on this, we consider three different
options.

In all the cases let $\uval, \vval \in \Val(\tau)$ and define $R = \{(\tval, \uval), (\fval,
\vval)\}$ and $S = \{(\uval, \tval), (\vval, \fval)\}$. Note that the cases don't depend on $\uval, \vval,
R$ or $S$.
\begin{itemize}
\item $\letapp{(\TAPP{v}{\B})}{\PAIR{\tval}{\fval}} \stepsto^* \tval$ but
  $\lnot\left(\letapp{(\TAPP{v}{\B})}{\PAIR{\tval}{\fval}} \stepsto^* \fval\right)$. In
  this case, we wish to show that $\letapp{(\TAPP{v}{\tau})}{\PAIR{\uval}{\vval}} \mustctxtequiv
  \uval$ and we again show this by showing that the two terms are must-CIU equivalent.

  Let $E\in\Stk(\tau)$ and assume
  $E\left[\letapp{(\TAPP{v}{\tau})}{\PAIR{\uval}{\vval}}\right]\Downarrow$. This implies there
  exists a $\nu < \omega_1$, such that
  $E\left[\letapp{(\TAPP{v}{\tau})}{\PAIR{\uval}{\vval}}\right]\Downarrow_{\nu}$. We must show
  $E[\uval]\Downarrow$. Suppose instead that $\lnot(E[\uval]\Downarrow)$. Then $\forall \beta <
  \omega_1, (E, (\FUN{x}{\IFT{x}{\Omega_{\B}}{x}})\,[]) \in \T{S_{\beta}}$. Instantiating
  (\ref{eq:lem:must:ex:converging-2}) with the above defined $S$ and any $\beta \geq \nu$ shows
  that 
  \[\left(\FUN{x}{\IFT{x}{\Omega_{\B}}{x}}\right)
  \left(\letapp{(\TAPP{v}{\B})}{\PAIR{\tval}{\fval}}\right)\Downarrow,\] but we have assumed that
  $\letapp{(\TAPP{v}{\B})}{\PAIR{\tval}{\fval}}]\stepsto^* \tval$.
  This therefore leads to a contradiction, stemming from the assumption
  that $E[\uval]$ may-diverges. This shows one direction of may-CIU approximation.

  For the other, again let $E\in\Stk(\tau)$ and now assume $E[\uval]\Downarrow$. It follows
  that for all $\beta < \omega_1$, $((\FUN{x}{\IFT{x}{x}{\Omega_{\B}}})\,[], E)\in\T{R_n}$. We
  now instantiate (\ref{eq:lem:must:ex:converging-1}) with our particular $R$ and $\nu <
  \omega_1$, such that $\left(\FUN{x}{\IFT{x}{x}{\Omega_{\B}}}\right)\,
    \left(\letapp{(\TAPP{v}{\B})}{\PAIR{\tval}{\fval}}\right)\Downarrow_\nu$. Such
  a $\nu$ exists since we have assumed that
  $\letapp{(\TAPP{v}{\B})}{\PAIR{\tval}{\fval}}$ must-converges and 
  $\letapp{(\TAPP{v}{\B})}{\PAIR{\tval}{\fval}} \stepsto^* \tval$ but it does not reduce
  to $\fval$ and so this implies
  $E\left[\letapp{(\TAPP{v}{\tau})}{\PAIR{\uval}{\vval}}\right]\Downarrow$, which concludes this
  part of the proof.

\item $\letapp{(\TAPP{v}{\B})}{\PAIR{\tval}{\fval}} \stepsto^* \fval$ but
  $\lnot\left(\letapp{(\TAPP{v}{\B})}{\PAIR{\tval}{\fval}} \stepsto^* \tval\right)$. In
  this case we show that $\letapp{(\TAPP{v}{\tau})}{\PAIR{\uval}{\vval}} \mustctxtequiv \vval$. The
  proof of this is completely analogous to the one for the previous case, so we omit the
  details here.

\item $\letapp{(\TAPP{v}{\B})}{\PAIR{\tval}{\fval}} \stepsto^* \fval$ \emph{and}
  $\letapp{(\TAPP{v}{\B})}{\PAIR{\tval}{\fval}} \stepsto^* \tval$. In
  this case, we wish to show that $\letapp{(\TAPP{v}{\tau})}{\PAIR{\uval}{\vval}} \mustctxtequiv
  \uval \OR \vval$. We again do this by showing must-CIU equivalence in two steps.

  Let $E\in\Stk(\tau)$ and assume
  $E\left[\letapp{(\TAPP{v}{\tau})}{\PAIR{\uval}{\vval}}\right]\Downarrow$. This implies there
  exists a $\nu < \omega_1$, such that
  $E\left[\letapp{(\TAPP{v}{\tau})}{\PAIR{\uval}{\vval}}\right]\Downarrow_\nu$. We must show
  $E[\uval\OR \vval]\Downarrow$. Suppose instead that $\lnot(E[\uval \OR \vval]\Downarrow)$. This
  implies $\lnot(E[\uval]\Downarrow)$ or $\lnot(E[\vval]\Downarrow)$ (or both). Without
  loss of generality suppose $\lnot(E[\uval]\Downarrow)$. This implies
  $\forall \beta < \omega_1, (E, (\FUN{x}{{\IFT{x}{\Omega_{\B}}{x}}})\,[]) \in \T{S_\beta}$. Instantiating
  (\ref{eq:lem:must:ex:converging-2}) with the above defined $S$ and any $\beta \geq \nu$ leads
  to a contradiction, since it implies that 
  \[(\FUN{x}{{\IFT{x}{\Omega_{\B}}{x}}})\,\letapp{(\TAPP{v}{\B})}{\PAIR{\tval}{\fval}}\Downarrow\]
  but since $\letapp{(\TAPP{v}{\B})}{\PAIR{\tval}{\fval}} \stepsto^* \tval$, this cannot be.

  For the other direction, let $E\in\Stk(\tau)$ and now assume $E[\uval \OR
  \vval]\Downarrow$. This implies that $E[\uval]\Downarrow$ and $E[\vval]\Downarrow$,
  which further implies that $\forall \beta < \omega_1, ([],E)\in\T{R_n}$. If we instantiate
  (\ref{eq:lem:must:ex:converging-1})
  with our particular $R$ and $\nu < \omega_1$, such that 
  \[ \letapp{(\TAPP{v}{\B})}{\PAIR{\tval}{\fval}}\Downarrow_\nu,
  \]
  (such a $\nu$ again exists by assumption) we have
  $E\left[\letapp{(\TAPP{v}{\tau})}{\PAIR{\uval}{\vval}}\right]\Downarrow$ which concludes
  the proof. \qedhere

\end{itemize}
\end{proof}

\begin{thm}
 \label{thm:parametricity}
If $v \in \Val\left(\ALL{\alpha}{\alpha\times\alpha\to\alpha}\right)$ then exactly one
of the following holds

\begin{itemize}
\item $v \mustctxtequiv \TFUN{\alpha}{\TAPP{\Omega}{(\alpha \times \alpha \to \alpha)}}$
\item $\forall \tau \in \Type, \forall \uval \in Val(\tau\times\tau),
  \letapp{(v\,\tau)}{\uval} \mustctxtequiv \TAPP{\Omega}{\tau}$
\item $\forall \tau\in\Type,\forall \uval, \vval \in \Val(\tau), 
  \letapp{(\TAPP{v}{\tau})}{\PAIR{\uval}{\vval}} \mustctxtequiv \uval$
\item $\forall \tau\in\Type,\forall \uval, \vval \in \Val(\tau), 
  \letapp{(\TAPP{v}{\tau})}{\PAIR{\uval}{\vval}} \mustctxtequiv \vval$
\item $\forall \tau\in\Type,\forall \uval, \vval \in \Val(\tau), 
  \letapp{(\TAPP{v}{\tau})}{\PAIR{\uval}{\vval}} \mustctxtequiv \uval \OR \vval$
\end{itemize}

If, further, $\forall \tau\in\Type, \exists v_{\tau} \in \Val(\tau\times\tau \to \tau)$, such
that $\TAPP{v}{\tau} \stepstopure v_{\tau}$ then one of the following holds
\begin{itemize}
\item $v \mustctxtequiv \TFUN{\alpha}{\FUN{x}{\TAPP{\Omega}{\alpha}}}$
\item $v \mustctxtequiv \TFUN{\alpha}{\FUN{x}{\syn{proj}_1\,x}}$
\item $v \mustctxtequiv \TFUN{\alpha}{\FUN{x}{\syn{proj}_2\,x}}$
\item $v \mustctxtequiv \TFUN{\alpha}{\FUN{x}{\syn{proj}_1\,x \OR \syn{proj}_2\,x}}$.
\end{itemize}
\end{thm}

\begin{proof}
The first part of the theorem only summarizes
Lemmas~\ref{lem:must:ex:div1},~\ref{lem:must:ex:div2},~\ref{lem:must:ex:converging}.

For the second part, we consider cases as in the first part and as all of them are
analogous, we only show the last one. Using the must-analogue of
Lemma~\ref{lem:may:extensionality-ctx-forall} we only need to show
\[
\forall \tau\in\Type, \TAPP{v}{\tau} \mustctxtequiv \FUN{x}{\syn{proj}_1\,x \OR
  \syn{proj}_2\,x}
\]
It is easy to show that $\TAPP{v}{\tau} \mustctxtequiv v_{\tau}$, but note that the fact
that $\TAPP{v}{\tau}$ reduces to $v_{\tau}$ using only pure reductions is crucial, as it
implies that this $v_{\tau}$ is the unique value of $\TAPP{v}{\tau}$. Using  transitivity of
$\mustctxtequiv$ it thus suffices to show $v_{\tau} \mustctxtequiv \FUN{x}{\syn{proj}_1\,x \OR
  \syn{proj}_2\,x}$. From the first part we have that 
\[\forall \uval, \vval \in \Val(\tau),
\letapp{(\TAPP{v}{\tau})}{\PAIR{\uval}{\vval}} \mustctxtequiv \uval \OR \vval\] and it is also immediate
that $(\FUN{x}{\syn{proj}_1\,x \OR \syn{proj}_2\,x})\,\PAIR{\uval}{\vval} \mustctxtequiv \uval\OR \vval$
which together imply
\[\forall \uval, \vval \in \Val(\tau), v_{\tau}\,\PAIR{\uval}{\vval} \mustctxtequiv
(\FUN{x}{\syn{proj}_1\,x \OR \syn{proj}_2\,x})\,\PAIR{\uval}{\vval}.\]

The canonical forms lemma shows that such pairs are the only possible values of type
$\tau\times\tau$ and so the must-analogue of Lemma~\ref{lem:may:fun-extensionality} implies that $v_{\tau}
\mustctxtequiv \FUN{x}{\syn{proj}_1\,x \OR \syn{proj}_2\,x}$, as required.
\end{proof}

Note that the example in Section~\ref{sect:maytheory}, used to demonstrate the lack of
extensionality for expression of function type, also demonstrates that
it is not the case that $v$ is contextually equivalent to one of the
functions listed in the above theorem without some further
restrictions, such as the one used for the second part of the theorem.

\section{Comparison to conference paper}
\label{sec:comparison-conf}

A preliminary version of this paper was presented at the 20th Annual Conference on Computer Science Logic (CSL'11), 12-15 September 2011.  The present version corrects some mistakes in the proof of syntactic minimal invariance in the earlier conference paper.  This was done by changing the counting of steps so that the only steps that count in the indexing of the logical relations are unfold-fold reductions.  That suffices for well-definedness of the logical relation, and means that the approximation relations are closed under pure zero-step reductions on the left and under pure arbitrary reductions on the right, which was implicitly used in the wrong proof in the conference paper.  This change means that the precise formulation of several lemmas have changed. 
Moreover, we have changed the parametricity example to a more
interesting one involving nondeterminism.

\section*{Acknowledgements}

We gratefully acknowledge the comments and suggestions from the referees. In particular
we thank them for discovering the problem with the earlier proof of syntactic minimal invariance
and for suggesting a better example of relational parametricity.


\bibliographystyle{abbrv}
\bibliography{local}

\begin{thebibliography}{10}

\bibitem{Agha:Mason:Smith:Talcott:97}
G.~Agha, I.~A. Mason, S.~F. Smith, and C.~L. Talcott.
\newblock A foundation for actor computation.
\newblock {\em J. Funct. Program.}, 7(1):1--72, 1997.

\bibitem{Appel:McAllester:01}
A.~W. Appel and D.~A. McAllester.
\newblock An indexed model of recursive types for foundational proof-carrying
  code.
\newblock {\em ACM Trans. Program. Lang. Syst.}, 23(5):657--683, 2001.

\bibitem{Apt:Plotkin:86}
K.~R. Apt and G.~D. Plotkin.
\newblock Countable nondeterminism and random assignment.
\newblock {\em J. ACM}, 33(4):724--767, 1986.

\bibitem{Birkedal:Harper:99}
L.~Birkedal and R.~Harper.
\newblock Relational interpretations of recursive types in an operational
  setting.
\newblock {\em Inf. Comput.}, 155(1-2):3--63, 1999.

\bibitem{Birkedal:Reus:Schwinghammer:Stovring:Thamsborg:Yang:11}
L.~Birkedal, B.~Reus, J.~Schwinghammer, K.~St{\o}vring, J.~Thamsborg, and
  H.~Yang.
\newblock Step-indexed {Kripke} models over recursive worlds.
\newblock In {\em POPL}, pages 119--132, 2011.

\bibitem{DiGianantonio:Honsell:Plotkin:95}
P.~{Di Gianantonio}, F.~Honsell, and G.~D. Plotkin.
\newblock Uncountable limits and the lambda calculus.
\newblock {\em Nord. J. Comput.}, 2(2):126--145, 1995.

\bibitem{DiGianantonio:Miculan:04}
P.~{Di Gianantonio} and M.~Miculan.
\newblock Unifying recursive and co-recursive definitions in sheaf categories.
\newblock In {\em FOSSACS}, pages 136--150, 2004.

\bibitem{BirkedalL:lslr-journal}
D.~Dreyer, A.~Ahmed, and L.~Birkedal.
\newblock Logical step-indexed logical relations.
\newblock {\em Logical Methods in Computer Science}, 7(2), 2011.

\bibitem{BirkedalL:stslr-conf}
D.~Dreyer, G.~Neis, and L.~Birkedal.
\newblock The impact of higher-order state and control effects on local
  relational reasoning.
\newblock In {\em ICFP}, pages 143--156, 2010.

\bibitem{johann:2010}
P.~Johann, A.~Simpson, and J.~Voigtl{\"a}nder.
\newblock A generic operational metatheory for algebraic effects.
\newblock In {\em LICS}, pages 209--218, 2010.

\bibitem{Laird:06}
J.~Laird.
\newblock Bidomains and full abstraction for countable nondeterminism.
\newblock In {\em FOSSACS}, pages 352--366, 2006.

\bibitem{Lassen:PhD}
S.~B. Lassen.
\newblock {\em Relational Reasoning about Functions and Nondeterminism}.
\newblock PhD thesis, University of Aarhus, 1998.

\bibitem{Lassen:Moran:99}
S.~B. Lassen and A.~Moran.
\newblock Unique fixed point induction for {McCarthy}'s amb.
\newblock In {\em MFCS}, pages 198--208, 1999.

\bibitem{Lassen:Pitcher:97}
S.~B. Lassen and C.~Pitcher.
\newblock Similarity and bisimilarity for countable non-determinism and
  higher-order functions.
\newblock {\em Electr. Notes Theor. Comput. Sci.}, 10, 1997.

\bibitem{Levy:06}
P.~B. Levy.
\newblock Infinitary {Howe}'s method.
\newblock In {\em CMCS}, pages 85--104, 2006.

\bibitem{Mason:Talcott:91}
I.~A. Mason and C.~L. Talcott.
\newblock Equivalence in functional languages with effects.
\newblock {\em J. Funct. Program.}, 1(3):287--327, 1991.

\bibitem{Pitts:96}
A.~M. Pitts.
\newblock Relational properties of domains.
\newblock {\em Inf. Comput.}, 127(2):66--90, 1996.

\bibitem{Pitts:05}
A.~M. Pitts.
\newblock Typed operational reasoning.
\newblock In B.~C. Pierce, editor, {\em Advanced Topics in Types and
  Programming Languages}, chapter~7, pages 245--289. MIT Press, 2005.

\bibitem{Pitts:10}
A.~M. Pitts.
\newblock Step-indexed biorthogonality: a tutorial example.
\newblock In {\em Modelling, Controlling and Reasoning About State}, Dagstuhl
  Seminar Proceedings, 2010.

\bibitem{Sabel:Schmidt-Schauss:08}
D.~Sabel and M.~Schmidt-Schau{\ss}.
\newblock A call-by-need lambda calculus with locally bottom-avoiding choice:
  context lemma and correctness of transformations.
\newblock {\em Math. Struct. Comp. Sci.}, 18(3):501--553, 2008.

\end{thebibliography}

\end{document}